\documentclass{article}
\usepackage{graphicx} % Required for inserting images
\usepackage{mathtools,amsmath,amsfonts,amsthm,amssymb}
\usepackage[utf8]{inputenc}
\usepackage{xcolor}
\usepackage{hyperref}
\usepackage{amsthm}
\usepackage{tocloft}
\usepackage{authblk}
\usepackage{tikz-cd}
\newtheorem{theorem}{Theorem}
\theoremstyle{definition}
\newtheorem{defn}[theorem]{Definition} 
\newtheorem{exmp}[theorem]{Example}
\newtheorem{lem}[theorem]{Lemma}
\newtheorem{cor}[theorem]{Corollary}
\newtheorem{rmk}[theorem]{Remark}
\newtheorem{prop}[theorem]{Proposition}

\newcommand{\Hom}{\mathrm{Hom}}
\newcommand{\im}{\mathop{\mathrm{im}}}
\newcommand{\coker}{\mathop{\mathrm{coker}}}

\newcommand{\id}{\mathrm{id}}

\newcommand*\colvec[3][]{
    \begin{pmatrix}\ifx\relax#1\relax\else#1\\\fi#2\\#3\end{pmatrix}
}

\usepackage{upgreek}
\usepackage{mathrsfs}
\usepackage{setspace}
\usepackage{graphicx}
\usepackage{subcaption}
\usepackage{slashed}
\newtheorem{lemma}{Lemma}[section]

\newtheorem{defobs}{Definition/Observation}

\newcommand{\beq}{\begin{equation}}
\newcommand{\eeq}{\end{equation}}
\newcommand{\eeeem}{\end{multline}}
\newcommand{\bem}{\begin{multline}}
\newcommand{\bqa} {\begin{eqnarray}}
\newcommand{\eqa} {\end{eqnarray}}

\newcommand{\bmul}{\begin{multline}}
\newcommand{\emul}{\end{multline}}

\DeclareMathOperator{\Id}{Id}

\DeclareMathOperator{\Ad}{Ad}

\newcommand{\CA}{{ \mathcal A}}

\newcommand{\CalC}{{\mathcal C}}

\newcommand{\ZZ}{{\mathbb Z}}
\newcommand{\RR}{{\mathbb R}}
\newcommand{\CC}{{\mathbb C}}

\newcommand{\SA}{{\mathscr A}}

\def \al {\alpha}

\usepackage{tikz}
\usetikzlibrary{decorations.markings}

\newtheorem{thm}{Theorem}

\newcommand{\A}{\mathcal A}

\newcommand{\R}{\mathbb R}

\newcommand\norm[1]{\left\lVert#1\right\rVert}

\usepackage[style=numeric-comp,sorting=none
]{biblatex}

\title{Categorifying Clifford QCA}
\author{Bowen Yang}
\affil{Center of Mathematical Sciences and Applications, Harvard University, Cambridge, Massachusetts 02138, USA}
\date{\today}

\begin{document}

\maketitle
\begin{abstract}
We provide a complete classification of Clifford quantum cellular automata (QCAs) on arbitrary metric spaces and any qudits (of prime or composite dimensions) in terms of algebraic \( L \)-theory. Building on the delooping formalism of Pedersen and Weibel, we reinterpret Clifford QCAs as symmetric formations in a filtered additive category constructed from the geometry of the underlying space. This perspective allows us to identify the group of stabilized Clifford QCAs, modulo circuits and separated automorphisms, with the Witt group of the corresponding Pedersen--Weibel category. Notably, because the Pedersen–Weibel category depends only on the large-scale (coarse) structure of the metric space, so too does the classification of Clifford QCAs. For Euclidean lattices, the classification reproduces and expands upon known results, while for more general spaces---including open cones over finite simplicial complexes---we relate nontrivial QCAs to generalized homology theories with coefficients in the \( L \)-theory spectrum. Our results do not depend on translation symmetry. However, we do outline extensions to QCAs with symmetry and discuss how these fit naturally into the \( L \)-theoretic framework.
\end{abstract}

\section{Introduction}

Quantum cellular automata (QCAs) are locality-preserving automorphisms of quantum many-body systems. As such, they serve as a broad mathematical framework encompassing quantum circuits, translations, and more exotic dynamical symmetries~\cite{freedman2020classification,gross2012index, haah2021clifford, haah2023nontrivial, haah2025topological, shirley2022three, chen2023exactly, fidkowski2024pumping, fidkowski2024qca}. Of particular interest are \emph{Clifford QCAs}, which preserve the structure of generalized Pauli operators and arise naturally in models of stabilizer codes, quantum error correction, and condensed matter physics. While Clifford QCAs are tractable from both physical and computational perspectives, their classification on general metric spaces remains a subtle and conceptually rich problem.

This paper provides a complete classification of Clifford QCAs in terms of algebraic \(L\)-theory, building on a construction due to Pedersen and Weibel~\cite{pedersen2006nonconnective, pedersen1989k, pedersen1982ki}. Their delooping formalism—originally developed to define negative \(K\)-groups—naturally accommodates the large-scale geometry of the underlying space, and enables us to express Clifford QCA classification as a form of generalized homology theory. Our main result realizes the classification group of Clifford QCAs as an algebraic \(L\)-group associated to a category of modules.

More precisely, we show that the group \(K(\Lambda, \mathbb{Z}_d)\) of nontrivial Clifford QCAs on a metric space \(\Lambda\) modulo trivial automorphisms is isomorphic to \(L^1(\mathcal{C}_\Lambda(\mathcal{A}), -1)\), where \(\mathcal{C}_\Lambda(\mathcal{A})\) is the Pedersen–Weibel category built from finitely generated, free \(\mathbb{Z}_d\)-modules and the geometry of \(\Lambda\). When \(\Lambda\) is Euclidean space, our classification generalizes known results~\cite{haah2023nontrivial,haah2025topological}; more generally, if \(\Lambda = O(X)\) is an open cone over a finite simplicial complex \(X\), we show that the classification is given by the homology of \(X\) with coefficients in the \(L\)-theory spectrum of \(\mathbb{Z}_d\).

Our perspective makes essential use of additive \(L\)-theory in the sense of Ranicki~\cite{ranicki1973algebraic,ranicki1973algebraicII,ranicki1989additive,ranicki1992lower,ranicki1992algebraic}, particularly the notion of quadratic and symmetric formations, which provides a concrete model for \(L\)-groups amenable to elementary definitions. By identifying Clifford QCAs with symmetric formations in the Pedersen–Weibel category, we connect physically motivated equivalence relations with homotopy-theoretic invariants. The approach presented here not only generalizes known periodicity phenomena (such as the fourfold periodicity of Clifford QCA phases over \(\mathbb{Z}_d\) for odd prime \(d\)), but also extends naturally to settings with symmetry and mixed qudit dimensions. We discussed how internal or crystallographic symmetries lead to modified classification groups involving equivariant module categories, and we define a coarse-graining procedure corresponding to direct limits over finite-index subgroups. %This unifies and extends earlier work on symmetry-respecting QCAs.

The structure of the paper is as follows. Section~2 reviews the Pedersen–Weibel construction and establishes its basic properties. In Section~3, we define Clifford QCAs and introduce the classification group \(K(\Lambda, \mathbb{Z}_d)\). Section~4 develops the necessary background in additive \(L\)-theory, and proves an \(L\)-theoretic version of the delooping theorem. Section~5 contains the main classification results, including a remark in the end on how invertible subalgebra fits into our framework. We conclude in Section~6 with a discussion of future directions and potential extensions. Appendix A reviews QCA and motivates an observation which greatly simplifies Clifford QCA. Appendix~B includes a detailed proof of a theorem due to Pedersen--Weibel~\cite{pedersen2006nonconnective}, whose underlying idea is central to this work.
The diagram below illustrates the relationships between the various concepts.

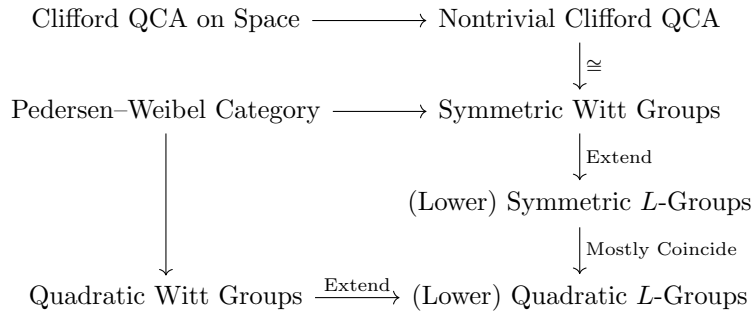
\begin{tikzcd}
     \text{Clifford QCA on Space} \arrow[r] &  {\text{Nontrivial Clifford QCA}} \arrow[d, "\cong"]\\
      {\text{Pedersen--Weibel Category}} \arrow[r]\arrow[dd]&  {\text{Symmetric Witt Groups}}\arrow[d, "\text{Extend}"]\\
      &  {\text{(Lower) Symmetric $L$-Groups}}\arrow[d, "\text{Mostly Coincide}"]\\
       {\text{Quadratic Witt Groups}}\arrow[r, "\text{Extend}"]&  {\text{(Lower) Quadratic $L$-Groups}}
\end{tikzcd}

\medskip

\noindent\textbf{Acknowledgments.} 
I thank Dan Freed and Mike Freedman for their many insights and generous support throughout the development of this work. I have benefited from conversations with Anton Kapustin about the nascent ideas behind this project since early 2022. I thank Agn\'es Beaudry and the two anonymous referees for carefully reading the manuscript and providing detailed feedback. I gratefully acknowledge support from Harvard CMSA and the Simons Foundation through Simons
Collaboration on Global Categorical Symmetries. I also thank Yu-An Chen, Roman Geiko, Jeongwan Haah, Mike Hopkins, Blazej Ruba, Wilbur Shirley, Nikita Sopenko, and Nathanan Tantivasadakarn for helpful discussions. 

Special thanks to Shmuel Weinberger for his guidance on surgery theory and for his time and hospitality during my visit to the University of Chicago. I owe my knowledge of
$L$-theory to Shmuel, though all errors are entirely mine. 

\section{The Pedersen--Weibel Delooping}
This section outlines a construction due to Pedersen and Weibel, which underpins a non-connective delooping of algebraic $K$-theory~\cite{pedersen1982ki, pedersen2006nonconnective} and provides an elegant framework for describing the $K$-theory homology of spaces~\cite{pedersen1989k}. We present a central theorem that may be viewed as a \emph{proto-theorem} for the classification of Clifford QCAs.

\subsection{Preliminaries}
\begin{defn}
    An additive category is said to be filtered if there is a filtration by subgroups
    $$F_0\Hom(A,B)\subset F_1\Hom(A,B)\subset \dots \subset F_k\Hom(A,B)\subset \dots$$
    on $\Hom(A,B)$ with $\Hom(A,B)=\bigcup_{k=1}^\infty F_k\Hom(A,B)$. Additionally, we require that 
    \begin{enumerate}
        \item $0_A$ and $1_A$ belong to $F_0$,
        \item composition of morphisms in $F_m$ and $F_n$ belongs to $F_{m+n}$
        \item projections $A\oplus B\rightarrow A$ and inclusions $A\rightarrow A\oplus B$ belong to $F_0$.
    \end{enumerate}
     A morphism $\phi$ is said to have degree $r$ if $\phi\in F_r\Hom(A,B)$.
\end{defn}
 
\begin{defn} \label{defn: P-W Cat} [Pedersen--Weibel]
Let $(X, \rho)$ be a metric space and $\mathcal{A}$ be a filtered additive category. We then define the filtered additive category $\mathcal{C}_X(\mathcal{A})$ as follows:
\begin{enumerate}
    \item An object $A$ of $\mathcal{C}_X(\mathcal{A})$ is a collection of objects $A(x)$ of $\mathcal{A}$, one for each $x \in X$, satisfying the condition that for each ball $B \subset X$, $A(x) \ne 0$ for only finitely many $x \in B$.
    \item A morphism $\phi : A \rightarrow B$ is a collection of morphisms $\phi^x_y : A(x) \rightarrow B(y)$ in $\mathcal{A}$ such that there exists $r$ depending only on $\phi$ so that
    \begin{enumerate}
        \item $\phi^x_y = 0$ for $\rho(x, y) > r$
        \item all $\phi^x_y$ are in $F_r \mathrm{Hom}(A(x), B(y))$
    \end{enumerate}
    We then say that $\phi$ has filtration degree $\le r$.
\end{enumerate}

Composition of $\phi : A \rightarrow B$ with $\psi : B \rightarrow C$ is given by $(\psi \phi)^x_z = \sum_{y \in X} \psi^y_z \phi^x_y$. Notice that the sum makes sense because the category is additive and because the sum will always be finite. 
\end{defn}

The metric space \( X \) may be discrete, such as \( (\mathbb{Z}^d, \norm{\bullet}_\infty) \) or a graph, or continuous, such as a (noncompact) Riemannian manifold. Write $\mathcal C_i(\CA)= \mathcal C_{\ZZ^i}(\CA)$ for $i\geq 0$, we record some basic facts from~\cite{pedersen1982ki, pedersen2006nonconnective}.

\begin{lemma} [Lemma 1.4 in~\cite{pedersen1989k}]
Let $X$ and $Y$ be metric spaces and give $X \times Y$ the max metric
$\rho_{X \times Y}((x_1,y_1),(x_2,y_2)) = \max(\rho_X(x_1,x_2), \rho_Y(y_1,y_2))$.
Then $C_{X \times Y}(\mathcal{A}) = C_X(C_Y(\mathcal{A}))$.
\end{lemma}

Note that this equivalence only holds if $\mathcal{C}_Y(\CA)$ is treated as a filtered additive category. In particular, morphisms in $C_X(C_Y(\mathcal{A}))$ is subject to the condition 2(b) in Definition~\ref{defn: P-W Cat}. If the filtration of $C_Y(\mathcal{A})$ is forgotten, then the metric of $Y$ no longer restricts morphisms in $C_X(C_Y(\mathcal{A}))$. 

\begin{prop}
        It is easy to see that $\mathcal C_\Lambda(\CA)\cong \CA$ when $\Lambda$ is bounded. By the lemma above, $$\mathcal{C}_i(\mathcal{C}_j(\CA))=\mathcal C_{i+j}(\CA).$$
\end{prop}

\begin{prop}
    Let $\R^n$ be the Euclidean space. There is an equivalence of filtered additive categories
    $$\mathcal{C}_n(\CA)\cong \mathcal{C}_{\R^n}(\CA).$$
\end{prop}
\begin{proof}
   Viewing $\ZZ^n$ as a metric subspace of $\RR^n$, $\mathcal C_n(\CA)$ is naturally a filtered additive full subcategory of $\CalC_{\R^n}(\CA).$ An object $(A(j_1, \dots, j_n))_{(j_1, \dots, j_n)\in \R^n}\in \CalC_{\R^n}(\CA)$ is isomorphic to $(A'(m_1, \dots, m_n))_{(m_1, \dots, m_n)\in \ZZ^n}$ given by $$A'(m_1, \dots, m_n):=\bigoplus_{\lfloor j_k \rfloor=m_k, \textbf{ for all } k}A(j_1, \dots, j_n),$$ which always has a finite number of nontrivial summands. 
   Therefore, the inclusion of $\CalC_n(\CA)$ is essentially surjective. 
\end{proof}
More generally, a coarse equivalence (see Section~1.4 of~\cite{nowak2023large}) between two metric spaces $X$ and $Y$ induces an equivalence between the filtered additive categories $\mathcal C_X(\Lambda)$ and $\mathcal C_Y(\Lambda)$. We do not develop the coarse-geometric framework needed to establish this result, as it lies within a broad and well-developed subject. Instead, we refer the reader to the standard texts~\cite{roe2003lectures, nowak2023large} for background and further details.

\subsection{Main Proto-Theorem}

 We set up and explain a key theorem in \cite{pedersen1982ki}.

\begin{defn}
    Let $\mathcal{A}$ be an additive category. The \emph{idempotent completion} (or \emph{Karoubi envelope}) of $\mathcal{A}$ is a category $\operatorname{Kar}(\mathcal{A})$ defined as follows:

\begin{itemize}
    \item Objects are pairs $(A, e)$, where $A$ is an object of $\mathcal{A}$ and $e: A \to A$ is an idempotent morphism in $\mathcal{A}$, i.e., $e^2 = e$.
    
    \item A morphism $\phi: (A, e) \to (B, f)$ in $\operatorname{Kar}(\mathcal{A})$ is a morphism $\phi: A \to B$ in $\mathcal{A}$ such that $\phi = f   \phi   e$. In particular, this implies $(A, 0)\cong (B, 0)$ for any $A, B\in \CA.$
    
    \item The filtration degree of $\phi$ is the smallest $k$ such that $\phi=f \psi e$ for some $\psi \in F_k\Hom_\CA(A, B)$.
\end{itemize}

The category $\operatorname{Kar}(\mathcal{A})$ is additive. By $A\mapsto (A, 1_A)$, the category $\mathcal{A}$ embeds fully and faithfully in $\operatorname{Kar}(\mathcal{A})$.
\end{defn}

\begin{defn}

Let $\mathcal{A}$ be an additive category. The \emph{Grothendieck group} $K_0(\mathcal{A})$ is defined as follows.

\begin{itemize}
  \item Consider the commutative monoid formed by isomorphism classes of objects in $\mathcal{A}$, with the operation induced by direct sum:
  \[
  [A] + [B] := [A \oplus B].
  \]
  \item The Grothendieck group $K_0(\mathcal{A})$ is the group completion of this monoid. That is, $K_0(\mathcal{A})$ is the abelian group generated by symbols $[A]$ for each object $A \in \mathcal{A}$, subject to the relation
  \[
  [A \oplus B] = [A] + [B] \quad \text{for all } A, B \in \mathcal{A}.
  \]
\end{itemize}

Elements of $K_0(\mathcal{A})$ can be represented as formal differences $[A] - [B]$ of isomorphism classes of objects in $\mathcal{A}$.

\end{defn}

\begin{defn}
Let $\mathcal{A}$ be an additive category. The \emph{algebraic \(K_1\)-group} of $\mathcal{A}$, denoted \(K_1(\mathcal{A})\), is defined as follows:

\medskip

We first consider the category $\operatorname{Aut}(\mathcal{A})$, whose objects are pairs \((A, \phi)\) where \(A\) is an object of \(\mathcal{A}\) and \(\phi : A \to A\) is an automorphism in \(\mathcal{A}\). A morphism \(f : (A, \alpha) \to (B, \beta)\) in \(\operatorname{Aut}(\mathcal{A})\) is a morphism \(f : A \to B\) in \(\mathcal{A}\) such that
\[
f   \alpha = \beta   f.
\]

Then \(K_1(\mathcal{A})\) is defined as the abelian group obtained from the group completion of the monoid of isomorphism classes in \(\operatorname{Aut}(\mathcal{A})\), modulo the relations:
\begin{enumerate}
    \item $[A, \alpha]+ [A, \alpha']= [A, \alpha  \alpha']$
    \item $[A, \alpha]+[C, \gamma]= [B, \beta]$ whenever there is a diagram with exact rows
    \[
\begin{array}{ccccccccc}
0 & \longrightarrow & A & \longrightarrow & B & \longrightarrow & C & \longrightarrow & 0 \\
  &                 & \downarrow{\scriptstyle \alpha} 
  &                 & \downarrow{\scriptstyle \beta} 
  &                 & \downarrow{\scriptstyle \gamma} 
  &                 &   \\
0 & \longrightarrow & A & \longrightarrow & B & \longrightarrow & C & \longrightarrow & 0
\end{array}.
\]

\end{enumerate}

\end{defn}

The definitions of algebraic $K_m$-groups for $m \geq 2$ are more involved and lie beyond the scope of this work. We refer the reader to Chapter IV of~\cite{weibel2013k} for precise definitions. %There are several different approaches to defining these groups, and their development occupies the entirety of that chapter.

The theorem below illustrates a specific instance of a proto-theorem that we seek to generalize across different physical settings.

\begin{thm} [Main Theorem in \cite{pedersen1982ki}] \label{thm:Main}
\begin{enumerate}
    \item For a filtered additive category $\CA$, there is a natural isomorphism 
$$ K_1(\mathcal{C}_{i+1}(\CA)) \cong K_0(\operatorname{Kar}(\mathcal{C}_{i}(\CA))).$$
\item If $\CA$ is the category of finitely generated, free $R$-modules (with trivial filtration), 
    $$ K_1(\mathcal{C}_{i+1}(\CA)) \cong K_0(\operatorname{Kar}(\mathcal{C}_{i}(\CA)))= K_{-i}(R),$$ where $K_{-i}(R)$ is the negative algebraic $K$ group of $R$ inductively defined to be the cokernel of the map $$K_{-i+1}(R[t]) \oplus K_{-i+1}(R[t^{-1}]) \to K_{-i+1}(R[t,t^{-1}]).$$ 

\end{enumerate}
\end{thm}
\begin{proof}
    See Appendix~B. Also see Section 4 of Chapter III in~\cite{weibel2013k} for detail on negative $K$-groups. 
\end{proof}
For those familiar with spectra, the second part follows from a more general theorem. We include it for completeness.

Let $X$ be a pointed subcomplex of the $n$-sphere $S^n$, and form the open cone $O(X)$ inside $\R^{n+1}$. Let $\mathbb K\CA$ be the nonconnective $K$-theory spectrum of $\CA$ established in Theorem A of~\cite{pedersen2006nonconnective}. The $\mathbb K\CA$-homology of $X$ is defined as 
$$ \mathbb K\CA_m(X):=\lim_{k\rightarrow \infty}\pi_{m+k}(\mathbb K_k\CA\wedge X).$$ 
\begin{thm}[Main Theorem in \cite{pedersen1989k}]\label{thm:main2}
     The $\mathbb K\CA$-homology is naturally isomorphic to the algebraic $K$-theory of the idempotent completion of $\mathcal{C}_{O(X)}(\CA)$, with a degree shift:
    $$ \mathbb K\CA_{m-1}(X)\cong K_m(\operatorname{Kar}(\mathcal{C}_{O(X)}(\CA))).$$
\end{thm}
Assuming this, we deduce the second isomorphism of Theorem \ref{thm:Main}. Let $X=S^{i-1}$. Then $O(X)=\R^i$ and 
\begin{align*}
    &K_m(\operatorname{Kar}(\mathcal{C}_{O(X)}(\CA)))\\
    %\cong &K_m(\operatorname{Kar}(\mathcal{C}_{i}(\CA)))\\
    \cong & \mathbb K\CA_{m-1}(S^{i-1})\\
    =&\pi_{m-i}(\mathbb K\CA),
\end{align*}
the stable homotopy group of $\mathbb K\CA$. If $\CA$ is the category of finitely generated, free $R$-modules (with trivial filtration) and $m=0$, the second part of Theorem \ref{thm:Main} follows.

\section{Clifford QCAs}

%\subsection{Preliminaries}
We present an unconventional definition of a Clifford QCA. It is equivalent to the usual version up to ambiguities that are irrelevant for classification purposes. See Lemma~\ref{lem: motivation} of Appendix~A for a proof. This alternative has several advantages, one of which is its close connection to the Pedersen–Weibel construction. See the rest of Appendix~A for a conventional account on QCAs and their classification program. The seminal papers~\cite{freedman2020classification} and~\cite{haah2021clifford} are also great resources on the topic. 
\begin{defobs}    \label{obs:keyobservation}
A Clifford QCA on a metric space $\Lambda$ with metric $\rho$ consists of the following data:
\begin{itemize}
    \item  A locally finite collection $P := (P(i))_{i \in \Lambda}$, where
$P(i) = \left(\mathbb{Z}_d \oplus \mathbb{Z}_d,
\begin{bmatrix}
0 & 1 \\
-1 & 0
\end{bmatrix}
\right)^{\oplus k_i}$.
Here, `locally finite' means that for any bounded subset $\Lambda' \subset \Lambda$, we have $P(i)=0$ for all but finitely many $i \in \Lambda'$.
    \item A collection of abelian group homomorphisms $(\alpha_j^i: P(i)\rightarrow P(j))_{i, j\in \Lambda}$ that are
    \begin{enumerate}
        \item symplectic: $\alpha^i_j$ pulls back the standard symplectic form, and
        \item invertible: there exists $(\beta_j^i: P(i)\rightarrow P(j))_{i, j\in \Lambda}$ such that
    $$(\beta \alpha)^i_l = \sum_{j \in \Lambda} \beta^j_l \alpha^i_j=\delta_{il}.$$
    
\end{enumerate}

\item A constant $r>0$ depending only on $\alpha$ such that $\alpha^i_j=0$ whenever $\rho(i,j)>r$.
\end{itemize}
\end{defobs}
Notice that the direct sum of two Clifford QCAs $(P, \alpha)$ and $(Q, \beta)$ is also a Clifford QCA $(P\oplus Q, \alpha \oplus \beta).$ When $P=Q$, we can compose $\alpha$ and $\beta$ to get another Clifford QCA $(P, \beta\alpha)$. Additionally, the inverse $(P, \alpha^{-1})$ is also a Clifford QCA. 
 
\begin{rmk}
    The parameters $k_i$ (number of qudits on $i\in \Lambda$) are allowed to vary. In particular, we choose not to put a uniform bound on them. 
\end{rmk}

\begin{defn}
    A Clifford QCA $(P, \alpha)$ is a \emph{single-layer Clifford circuit} if there is a partition of $\Lambda$ into disjoint regions $\Lambda=\bigcup_{s}\Lambda_s$, with $\mathrm{diam}(\Lambda_s)$ uniformaly bounded, such that
    \begin{itemize}
        \item $\alpha_j^i=0$ whenever, $i\in \Lambda_s$, $j\in \Lambda_{s'}$ for $s\ne s',$
        \item $(\alpha_j^i)_{i, j \in \Lambda_s}: \bigoplus_{k\in \Lambda_s} P(k) \rightarrow \bigoplus_{k\in \Lambda_s} P(k)$ is a symplectic automorphism for each $s\in \Lambda$. 
        
    \end{itemize}
    It is a \emph{Clifford circuit} if $\alpha$ is the composition of finitely many single-layer Clifford circuits. 
\end{defn}

\begin{defn}
A Clifford QCA $(P,\alpha)$ is said to be \emph{separated} if it preserves each summand of
$(P(i)=P^+(i)\oplus P^-(i)=\mathbb Z_d^{k_i}\oplus \mathbb Z_d^{k_i})_{i\in\Lambda}$.
That is, $\alpha^i_j(P^\pm(i))\subset P^\pm(j)$ for all $i,j\in\Lambda$.

The terminology “separated” reflects the fact that a separated QCA $(P,\alpha)$ decomposes as $\alpha=\alpha^+\oplus\alpha^-$, where $\alpha^\pm$ is the restriction to $P^\pm=(P^\pm(i))_{i\in\Lambda}$.
%Moreover, due to the symplectic condition, $\alpha$ is completely determined by either $\alpha^+$ or $\alpha^-$.
\end{defn}
\begin{defn}
    The classification of Clifford QCAs on $\Lambda$ is an abelian group $K(\Lambda, \ZZ_d)$ generated by Clifford QCAs $[P, \alpha]$ satisfying relations below:
    \begin{itemize}
        \item $[P, \alpha]+[P, \beta]=[P, \beta \alpha],$
        \item $[P, \alpha]+[Q, \beta]=[P\oplus Q, \alpha \oplus \beta],$
        \item $[P, \alpha]=0$ if $(P, \alpha)$ is separated,
        \item $[P, \alpha]=0$ if $(P, \alpha)$ is a Clifford circuit. 
    \end{itemize}
    In particular, we deduce that $$-[P, \alpha]=[P, \alpha^{-1}].$$
\end{defn}
\begin{rmk}
    Colloquially, $$K(\Lambda, \ZZ_d)=\frac{\text{Stabilized Clifford QCAs on }\Lambda}{\text{Stabilized Clifford circuits and separated QCAs on } \Lambda}.$$ 
    We follow the convention of~\cite{freedman2020classification} in deeming circuits to be trivial for the purposes of QCA classification. In particular, Theorem~2.3 of~\cite{freedman2020classification} shows that, after stabilization, QCAs modulo circuits form an abelian group. Our motivation for also trivializing separated QCAs is partly based on the view that QCAs which merely shift degrees of freedom in space should be regarded as trivial. Nevertheless, we later realized that the Pedersen--Weibel formalism is sufficiently robust to capture separated QCAs as well. Indeed, a separated QCA $(P, \al)$ gives rise to a class $[P^+, \al^+]$ in $K_1(\CalC_\Lambda(\CA))$, where $\mathcal{A}$ denotes the additive category of finitely generated, free $\mathbb{Z}_d$-modules. Thus, it is not surprising that stabilized separated QCAs, modulo separated Clifford circuits, are classified by the algebraic $K_1$-group $K_1(\CalC_\Lambda(\CA))$. Notably, for Euclidean space $\mathbb{R}^n$ with $n>0$, this classification group vanishes. See Section~6 of~\cite{sun2025clifford}, especially discussion below Theorem~2, for detail.
\end{rmk}
If we temporarily forget the standard symplectic form, $P$ is an object of 
$\mathcal{C}_\Lambda(\mathcal{A})$, where $\mathcal{A}$ denotes the additive category of finitely generated, free $\mathbb{Z}_d$-modules. 
In particular, the data of a Clifford QCA $\alpha^i_j$ defines an automorphism of $P$.

Thus, a Clifford QCA $[P,\alpha]$ canonically determines a class
$[P,\alpha] \in K_1(\mathcal{C}_\Lambda(\mathcal{A}))$.
This naturally raises the question of whether the full classification of Clifford QCAs can be captured by a refinement of algebraic $K$-theory that incorporates the symplectic structure.
 One naive attempt is to restrict ourselves to a subcategory of ``symplectic objects'' with morphisms that preserve symplectic structures. However, since sums of symplectic maps of modules are not necessarily symplectic, such a category would fail to be additive. Fortunately, this strategy succeeds if we use $L$-theory instead of $K$-theory.

\section{Additive L-theory}

Following \cite{ranicki1989additive}, we define both quadratic and symmetric $L$-theory for additive categories with involution. The section culminates in an $L$-theoretic analog of Theorem \ref{thm:Main}. We mention that \cite{ranicki1973algebraic, ranicki1989additive} also includes an alternative definition of $L$-theory using Poincar\'e complexes. 

\subsection{Preliminaries}

\begin{defn}

An \emph{involution} on an additive category \( \mathcal{A} \) is an additive contravariant functor
\[
    * : \mathcal{A}^{\mathrm{op}} \to \mathcal{A}
\]
together with a natural isomorphism
\[
    \eta: \mathrm{id}_{\mathcal{A}} \xrightarrow{\sim} * \circ *^{\mathrm{op}}
\]
such that, for every object \( A \in \mathcal{A} \),  \( (\eta_{A})^* \circ \eta_{A^*} = \mathrm{id}_{A^*} \), i.e., the double dual is naturally isomorphic to the identity functor and the involution is of order two up to isomorphism.

We denote such a structure by the pair \( (\mathcal{A}, *) \).
\end{defn}
\begin{exmp}
Let \( R \) be an associative ring equipped with a ring involution
\[
\overline{(\cdot)} : R \to R, \quad \overline{ab} = \bar{b}\bar{a}, \quad \overline{\bar{a}} = a.
\]
Let \( \mathcal{A} \) be the additive category whose objects are finitely generated, free left \( R \)-modules, and whose morphisms are \( R \)-linear maps.

Define a contravariant functor
\[
* : \mathcal{A}^{\mathrm{op}} \to \mathcal{A}
\]
as follows:

\begin{itemize}
  \item For each object \( M \in \mathcal{A} \), let
  \[
  M^* := \mathrm{Hom}_R(M, R),
  \]
  with the structure of a left \( R \)-module defined by
  \[
  (r \cdot f)(m) := f(m) \cdot \bar{r} \quad \text{for } r \in R, f \in M^*, m \in M.
  \]
  \item For a morphism \( f: M \to N \), define
  \[
  f^* : N^* \to M^*, \quad f^*(\phi) := \phi \circ f.
  \]
\end{itemize}

There is a natural isomorphism \( \eta_M : M \to M^{**} \), given by
\[
\eta_M(m)(\phi) := \overline{\phi(m)} \quad \text{for } m \in M, \phi \in M^*.
\]

Thus, \( (\mathcal{A}, *) \) is an additive category with involution.
\end{exmp}

\begin{defn}

Let \( \mathcal{A} \) be an additive category with involution. For objects \( M, N \in \mathcal{A} \), define a duality isomorphism
\[
T_{M,N} : \operatorname{Hom}_{\mathcal{A}}(M, N^*) \to \operatorname{Hom}_{\mathcal{A}}(N, M^*), \quad \psi \mapsto \psi^*\circ \eta_N.
\]
%such that
%\[
%T_{N,M} \circ T_{M,N} = \operatorname{id} : \operatorname{Hom}_{\mathcal{A}}(M, N^*) \to \operatorname{Hom}_{\mathcal{A}}(M, N^*).
%\]

In particular, when \( M = N \), we obtain a duality involution
\[
T := T_{M,M} : \operatorname{Hom}_{\mathcal{A}}(M, M^*) \to \operatorname{Hom}_{\mathcal{A}}(M, M^*), \quad \psi \mapsto \psi^*\circ \eta_M.
\]
\end{defn}
\begin{defn}

For \( \varepsilon = \pm 1 \) and \( M \in \mathcal{A} \), define the \( \varepsilon \)-duality involution
\[
T_\varepsilon := \varepsilon T : \operatorname{Hom}_{\mathcal{A}}(M, M^*) \to \operatorname{Hom}_{\mathcal{A}}(M, M^*), \quad \psi \mapsto \varepsilon \psi^*\circ \eta_M.
\]

\end{defn}

\begin{exmp}
    Let $R$ be an associative ring equipped with a ring involution $\overline{(\cdot)}: R\rightarrow R$. Let $\CA$ be the category of finitely generated, free left $R$-modules. For any $M, N\in \CA$ and $\psi \in \operatorname{Hom}_{\mathcal{A}}(M, N^*)$, we have $\psi^*(y)(x)=\overline{\psi(x)(y)}$ for all $x\in M, y\in N.$ If $M=N$, 
 by writing $\psi(x)(y)=:\langle x, y\rangle\in R$
$$(T_\varepsilon(\psi)(x))(y)= \varepsilon \overline{\langle y, x\rangle},$$ for all $x, y\in M.$

\end{exmp}

\subsection{Quadratic L-theory}
\begin{defn}
    
Let \( \mathcal{A} \) be an additive category with involution. An \( \varepsilon \)-quadratic form in \( \mathcal{A} \) is a pair \( (M, \psi) \), where \( M \in \mathcal{A} \) is an object together with an element
\[
\psi \in Q_\varepsilon(M) := \operatorname{coker}(1 - T_\varepsilon : \operatorname{Hom}_{\mathcal{A}}(M, M^*) \to \operatorname{Hom}_{\mathcal{A}}(M, M^*)).
\]

The form \( (M, \psi) \) is said to be \emph{non-singular} if the morphism
\[
(1 + T_\varepsilon)\psi = \psi + \varepsilon \psi^* : M \to M^*
\]
is an isomorphism in \( \mathcal{A} \).
\end{defn} 
\begin{defn}
    
A morphism of \( \varepsilon \)-quadratic forms
\[
f : (M, \psi) \to (M', \psi')
\]
is a morphism \( f : M \to M' \) in \( \mathcal{A} \) such that
\[
f^* \psi' f = \psi \in Q_\varepsilon(M).
\]

\end{defn} 

\begin{defn}

 Let \( (M, \psi) \) be a non-singular \( \varepsilon \)-quadratic form. A \emph{Lagrangian} in \( (M, \psi) \) is a morphism of forms
\[
i : (L, 0) \to (M, \psi)
\]
such that there exists a split exact sequence in \( \mathcal{A} \)
\[
0 \to L \xrightarrow{i} M \xrightarrow{i^*(\psi + \varepsilon \psi^*)} L^* \to 0.
\]
\end{defn}

\begin{defn}
    
Let \( L \in \mathcal{A} \). The \emph{hyperbolic \( \varepsilon \)-quadratic form} on \( L \) is the non-singular \( \varepsilon \)-quadratic form
\[
H_\varepsilon(L) = \left(L \oplus L^*, \begin{bmatrix} 0 & 1 \\ 0 & 0 \end{bmatrix} \right),
\]
with a Lagrangian defined by the morphism of forms
\[
i = \begin{bmatrix} 1 \\ 0 \end{bmatrix} : (L, 0) \to H_\varepsilon(L).
\]

\end{defn}
\begin{defn}
    
Let \( (M, \psi) \), \( (M', \psi') \) be non-singular \( \varepsilon \)-quadratic forms with Lagrangians \( L \), \( L' \), respectively. An isomorphism
\[
f : (M, \psi) \to (M', \psi')
\]
\emph{sends \( L \) to \( L' \)} if there exists an isomorphism \( e \in \operatorname{Hom}_{\mathcal{A}}(L, L') \) such that
\[
i' e = f i : L \to M',
\]
in which case the following is a morphism of split exact sequences:
\[
\begin{array}{ccccccccc}
0 & \to & L & \xrightarrow{i} & M & \xrightarrow{i^*(\psi + \varepsilon \psi^*)} & L^* & \to & 0 \\
  &     & \downarrow e &       & \downarrow f & & \downarrow (e^*)^{-1} & & \\
0 & \to & L' & \xrightarrow{i'} & M' & \xrightarrow{i'^*(\psi' + \varepsilon \psi'^*)} & L'^* & \to & 0
\end{array}
\]
\end{defn}
\begin{prop}\label{prop:normalization}
    
An \( \varepsilon \)-quadratic form \( (M, \psi) \) admits a Lagrangian \( L \) if and only if it is isomorphic to \( H_\varepsilon(L) \).

\end{prop} 
\begin{proof}
 An isomorphism of forms \( f : H_\varepsilon(L) \to (M, \psi) \) determines a Lagrangian \( L \) of \( (M, \psi) \) with
\[
i : L \xrightarrow{\left[ \begin{smallmatrix} 1 \\ 0 \end{smallmatrix} \right]} L \oplus L^* \xrightarrow{f} M.
\]

Conversely, suppose that \( (M, \psi) \) has a Lagrangian \( L \), and let \( i : L \to M \) be the inclusion. Choose a splitting morphism \( j \in \mathrm{Hom}_A(L^*, M) \) for the split exact sequence
\[
0 \to L \xrightarrow{i} M \xrightarrow{i^*(\psi + \varepsilon \psi^*)} L^* \to 0,
\]
so that
\[
i^*(\psi + \varepsilon \psi^*)j = 1 \in \mathrm{Hom}_A(L^*, L^*).
\]

For any \( k \in \mathrm{Hom}_A(L^*, L) \), define another splitting
\[
j' = j + ik : L^* \to M
\]
such that
\[
\begin{aligned}
j'^* \psi j' &= j^* \psi j + k^* i^* \psi i k + k^* i^* \psi j + j^* \psi i k \\
            % &= j^* \psi j + k^* i^* \psi i k + k^* (1-\varepsilon i^*\psi^*j)+ j^* \psi i k\\
            %&= j^* \psi j + k^* i^* \psi i k + k^* i^* \psi j + j^* \psi i k \\
             &= j^* \psi j + k^* \in Q_\varepsilon(L^*).
\end{aligned}
\]
The last equality follows from $i^*\psi i=0 \in Q_\varepsilon(L)$ by definition, and
$$k^*i^*\psi j +j^*\psi ik=k^*(1-\varepsilon i^*\psi^*j) +j^*\psi ik=k^* +(1-T_\varepsilon) j^*\psi ik.$$

For appropriate choice of $k$, the splitting \( j' : L^* \to M \)  is an inclusion of a Lagrangian, with
$j'^* \psi j' = 0 \in Q_\varepsilon(L^*).$ Then, we have that
\[
i \oplus j' : H_\varepsilon(L) \to (M, \psi)
\]
is an isomorphism of \( \varepsilon \)-quadratic forms. 
\end{proof}

\begin{defn}
    
The \emph{Witt group of \( \varepsilon \)-quadratic forms} \( W_\varepsilon(\mathcal{A}) \) is the abelian group generated by one generator \( (M, \psi) \) for each isomorphism class of non-singular \( \varepsilon \)-quadratic forms in \( \mathcal{A} \), subject to the following relations:

\begin{itemize}
  \item[(i)] \( (M, \psi) + (M', \psi') = (M \oplus M', \psi \oplus \psi') \),
  \item[(ii)] \( H_\varepsilon(L) = 0 \).
\end{itemize}
\end{defn} 

\begin{defn}
    
A \emph{non-singular \( \varepsilon \)-quadratic formation} in \( \mathcal{A} \) is a quadruple \( (M, \psi; F, G) \), consisting of a non-singular \( \varepsilon \)-quadratic form \( (M, \psi) \) together with an ordered pair of Lagrangians \( (F, G) \).
\end{defn} 

\begin{defn}

\begin{itemize}
  \item[(i)] An \emph{isomorphism of formations} in \( \mathcal{A} \)
  \[
  f : (M, \psi; F, G) \to (M', \psi'; F', G')
  \]
  is an isomorphism of forms \( f : (M, \psi) \to (M', \psi') \) which sends \( F \) to \( F' \) and \( G \) to \( G' \).

  \item[(ii)] A \emph{stable isomorphism of formations} in \( \mathcal{A} \)
  \[
  [f] : (M, \psi; F, G) \to (M', \psi'; F', G')
  \]
  is an isomorphism of formations
  \[
  f : (M, \psi; F, G) \oplus (H_\varepsilon(P); P, P^*) \to (M', \psi'; F', G') \oplus (H_\varepsilon(P'); P', P'^*)
  \]
  for some objects \( P, P' \in \mathcal{A} \).
\end{itemize}
    
\end{defn}

\begin{defn}
The \emph{Witt group of \( \varepsilon \)-quadratic formations} \( M_\varepsilon(\mathcal{A}) \) is the abelian group generated by one generator
\[
(M, \psi; F, G)
\]
for each stable isomorphism class of non-singular \( \varepsilon \)-quadratic formations in \( \mathcal{A} \), subject to the following relations:

\begin{itemize}
  \item[(i)] \( (M, \psi; F, G) + (M', \psi'; F', G') = (M \oplus M', \psi \oplus \psi'; F \oplus F', G \oplus G') \),
  \item[(ii)] \( (M, \psi; F, G) + (M, \psi; G, H) = (M, \psi; F, H) \).
\end{itemize}

The inverses in \( M_\varepsilon(\mathcal{A}) \) are given by
\[
-(M, \psi; F, G) = (M, \psi; G, F)  \in M_\varepsilon(\mathcal{A}).
\]
\end{defn} 
We discuss formations in the zero class of $M_\varepsilon(\CA)$.
\begin{defn}
    Two Lagrangians $F$ and $G$ of $(M, \psi)$ are called \emph{complementary} if $F\cap G=\{0\}$ and $F+G=M$.
\end{defn}
\begin{prop} [Lemma 9.13 in~\cite{luck2024surgery}]
    Given a formation $(M, \psi; F, G)$, if $F$ and $G$ are complementary, $(M, \psi; F, G)=0\in M_\varepsilon(\CA).$
\end{prop}
\begin{proof}

The inclusions of \( F \) and \( G \) into \( M \) induce an isomorphism \( h \colon F \oplus G \to M \). Choose a representative \( \psi \colon M \to M^* \) of \( \psi \in Q_\varepsilon(M) \). Then we can write
\[
h^* \circ \psi \circ h = 
\begin{pmatrix}
a & b \\
c & d
\end{pmatrix}
\colon F \oplus G \to (F \oplus G)^* = F^* \oplus G^*.
\]

Let \( \psi' \in Q_\varepsilon(F \oplus G) \) be the class of \( h^* \circ \psi \circ h \). Then \( h \) is an isomorphism of non-singular \( \varepsilon \)-quadratic forms \( (F \oplus G, \psi') \to (M, \psi) \), and \( F \) and \( G \) are lagrangians in \( (F \oplus G, \psi') \). Hence, the isomorphism \( (1 + \varepsilon T)(\psi') \colon F \oplus G \to F^* \oplus G^* \) is given by
\[
\begin{pmatrix}
a & b \\
c & d
\end{pmatrix}
+ \varepsilon 
T\left(
\begin{pmatrix}
a & b \\
c & d
\end{pmatrix}
\right)
=
\begin{pmatrix}
a + \varepsilon a^* & b + \varepsilon c^* \\
c + \varepsilon b^* & d + \varepsilon d^*
\end{pmatrix}
=
\begin{pmatrix}
0 & e \\
f & 0
\end{pmatrix}
\]
for some \( e \colon G \to F^* \) and \( f \colon F \to G^* \). The diagonal terms are $0$ because $F$ and $G$ are Lagrangian. Hence \( e = b + \varepsilon \cdot c^* \colon G \to F^* \) is an isomorphism. Define an isomorphism
\[
u = 
\begin{pmatrix}
1 & 0 \\
0 & e
\end{pmatrix}
\colon F \oplus G \xrightarrow{\sim} F \oplus F^*.
\]

One easily checks that \( u \) defines an isomorphism of formations
\[
u \colon (F \oplus G, \psi'; F, G) \xrightarrow{\sim} (H_\varepsilon(F); F, F^*).
\]
\end{proof}

\begin{prop}
    A formation $(M, \psi; F, G)=0\in M_\varepsilon(\CA)$ if there exists a Lagrangian $L$ complementary to both $F$ and $G$. In particular, $(M, \psi; F, F)=0\in M_\varepsilon(\CA).$ 
\end{prop}
\begin{proof}
    Note that $$(M, \psi; F, G)=(M, \psi; F, L)+(M, \psi; L, G),$$ and the result follows from the previous proposition. 
\end{proof}
\begin{defn}
    We call a formation $(M, \psi; F, G)$ \emph{trivial} if $F$ and $G$ are complementary. We call it \emph{elementary} or \emph{boundary} if there exists a Lagrangian $L$ complementary to both $F$ and $G$.
\end{defn}
We state without proof a theorem by Ranicki regarding the zero class in $M_\varepsilon(\CA)$.
\begin{thm}[Theorem 2.3 in~\cite{ranicki1973algebraic}]
    If a formation belongs to the zero class, it is stably isomorphic to the direct sum of a trivial formation and an elementary formation. A formation \( (M, \psi; F, G) \) is isomorphic to the sum of a trivial formation and an elementary formation if and only if there is a lagrangian complement \( \widehat{F} \) for \( F \) with the property that the projection
\[
\pi \colon M = F \oplus \widehat{F} \xrightarrow{\begin{pmatrix} 1 & 0 \end{pmatrix}} F,
\]
satisfies the following:

\begin{itemize}
    \item \( \pi(G) \) is finitely generated and free,
    \item there exists a finitely generated, free submodule \( L \subseteq F \) such that \( \pi(G) \oplus L = F \).
\end{itemize}

\noindent Moreover, the roles of \( F \) and \( G \) can be interchanged.
\end{thm}

%Using this, we prove a useful lemma that must be known, though we have not found it stated explicitly in the literature. 

% \begin{proof}
% By the Chinese Remainder Theorem, it suffices to prove it for a prime power $d=p^r$. For a formation \( (M, \psi; F, G) \), find a lagrangian complement \( \widehat{F} \) for \( F \) define \[
% \pi \colon M = F \oplus \widehat{F} \xrightarrow{\begin{pmatrix} 1 & 0 \end{pmatrix}} F.
% \]

% \end{proof}
To define quadratic \( L \)-theory, we first rename these groups as follows:
\begin{align*}
    W_{\varepsilon}(\mathcal{A}) := L_0(\mathcal{A},\varepsilon), \\
M_{\varepsilon}(\mathcal{A}) := L_1(\mathcal{A}, \varepsilon).
\end{align*}

%There are definitions~\cite{ranicki1973algebraic,ranicki1989additive} which extends to $L_n(\CA, \varepsilon)$ with $n\in \ZZ$. 
It turns out the groups $L_n(\CA, \varepsilon)$ are periodic with period four. Moreover, there exist isomorphisms $$L_n(\CA, \varepsilon)\cong L_{n+2}(\CA, -\varepsilon).$$ Write $$L_n(\CA):=L_n(\CA, +),$$ we now have 
\begin{align*}
    &L_0(\CA)= W_{+}(\CA),\\
    &L_1(\CA)= M_{+}(\CA),\\
    &L_2(\CA)= W_{-}(\CA),\\
    &L_3(\CA)= M_{-}(\CA),\\
    &L_{i+4}(\mathcal{A}) = L_i(\mathcal{A}).
\end{align*}
 %Somewhat unexpectedly, Witt groups constitute a complete definition of the quadratic $L$-groups. 
\begin{rmk}
It is worth pointing out that there exists other definitions of $L$-groups which makes the four-fold periodicity appear less artificial.
\end{rmk}

\subsection{Symmetric L-theory}
\begin{defn}
    
Let \( \mathcal{A} \) be an additive category with involution. An \( \varepsilon \)-symmetric form in \( \mathcal{A} \) is a pair \( (M, \psi) \), where \( M \in \mathcal{A} \) is an object together with an element
\[
\psi \in Q^\varepsilon(M) := \operatorname{ker}(1 - T_\varepsilon : \operatorname{Hom}_{\mathcal{A}}(M, M^*) \to \operatorname{Hom}_{\mathcal{A}}(M, M^*)).
\]

The form \( (M, \psi) \) is said to be \emph{non-singular} if the morphism
\[
 \psi : M \to M^*
\]
is an isomorphism in \( \mathcal{A} \).
\end{defn} 
\begin{rmk}
    If $\CA$ is a category of modules over a ring where $2$ is invertible, symmetric and quadratic forms coincide. 
\end{rmk}
Repeating the previous discussion replacing quadratic forms with symmetric forms. We arrive at the Witt groups of $\varepsilon$-symmetric forms $W^\varepsilon$ and $\varepsilon$-symmetric formations $M^\varepsilon.$ 

To define symmetric \( L \)-theory, we rename these groups as follows:
\begin{align*}
    W^{\varepsilon}(\mathcal{A}) := L^0(\mathcal{A},\varepsilon), \\
M^{\varepsilon}(\mathcal{A}) := L^1(\mathcal{A}, \varepsilon).
\end{align*}

It is possible to extend them into $L^n(\CA, \varepsilon)$ with $n\in \ZZ$~\cite{ranicki1973algebraic,ranicki1989additive}. Contrary to the case in quadratic $L$-theory, the symmetric $L$-groups, denoted by $L^n(\CA, \varepsilon)$, are not always four-periodic. Also, in general, $$L^n(\CA, \varepsilon)\ncong L^{n+2}(\CA, -\varepsilon).$$ Therefore, Witt groups are not sufficient in defining all symmetric $L$-groups.

Symmetric \( L \)-theory plays an essential, though temporary, role in the classification of Clifford QCAs. It is essential because the group \( K(\Lambda, \mathbb{Z}_d) \) is directly related to symmetric \( L \)-theory (see Theorem~\ref{thm: injectivity}). It is temporary because, in most cases relevant to our classification, they coincide with quadratic \( L \)-groups (see Proposition~\ref{prop:quad2sym}).

\begin{rmk}
We follow the conventions of~\cite{ranicki1992lower} in our use of the notation \( L^n \). In particular, the superscript \( n \) does not indicate cohomological indexing.
\end{rmk}

\subsection{Instantiation of the Proto-Theorem}

We also require the so-called lower quadratic and symmetric $L$-groups defined by Ranicki in Sections 17 and 19 in~\cite{ranicki1992lower}. 

Recall the Pedersen--Weibel construction which produces a filtered additive category $\mathcal C_X(\CA)$ from a metric space $X$ and a filtered additive category $\CA$. We denote $\mathcal C_{\ZZ^i}(\CA)$ by $\mathcal C_i(\CA)$. If $\CA$ comes with an involution, $\mathcal C_X(\CA)$ is a filtered additive category with the involution defined point-wise. 
\begin{defn}
    An object $(A(k))_{k\in \ZZ}\in \mathcal C_1(\CA)$ consists of $A(k)\in \CA.$ Define the Laurent extension category, usual denoted $\CA[\ZZ]$ or $\CA[z, z^{-1}]$, to be the subcategory with uniform objects, i.e., $A(k)=A\in \CA$ for all $k\in \ZZ$, and $\ZZ$-equivariant morphisms.  An object in $\CA[z, z^{-1}]$ is denoted by 
    $$\sum_{k\in \ZZ}Az^k.$$ For $A\in \CA$, $i: A \mapsto \sum_{k\in \ZZ}Az^k$ is an inclusion functor $\CA\rightarrow \CA[z, z^{-1}].$
\end{defn}
We define the multivariable Laurent extension $\mathcal{A}[\mathbb{Z}^m]$ by a straightforward generalization.
\begin{rmk}
   If \( \mathcal{A} \) is the category of free modules, then the Laurent extension \( \mathcal{A}[\mathbb{Z}^m] \) is equivalent to the category of free \( R[\mathbb{Z}^m] \)-modules. Indeed, a free $R[\mathbb{Z}^m]$-module $M$ is isomorphic (as $R$-modules) to 
   $\bigoplus_{\mathbb{Z}^m}M_0$ for some free $R$-module $M_0$ equipped with a $\mathbb{Z}^m$-action. 

\end{rmk}
%The lower quadratic (resp. symmetric) \( L \)-groups \( L_n^{\langle -m \rangle}(\mathcal{A}) \) (resp. \( L^n_{\langle -m \rangle}(\mathcal{A}) \)) for \( m \geq 0 \) are defined for any additive category with involution \( \mathcal{A} \).

\begin{defn}
    Let \( \mathcal{A} \) be an additive category with involution. For each \( m \geq 0 \) and \( n \in \mathbb{Z} \), the \emph{lower quadratic} and \emph{lower symmetric} \( L \)-groups of \( \mathcal{A} \) are defined by
    \[
    L_n^{\langle -m \rangle}(\mathcal{A}) := \operatorname{coker}\big(i_! : L_{n+1}^{\langle -m+1 \rangle}(\mathcal{A}) \longrightarrow L_{n+1}^{\langle -m+1 \rangle}(\mathcal{A}[z, z^{-1}])\big),
    \]
    and
    \[
    L^n_{\langle -m \rangle}(\mathcal{A}, \varepsilon) := \operatorname{coker}\big(i_! : L^{n+1}_{\langle -m+1 \rangle}(\mathcal{A}, \varepsilon) \longrightarrow L^{n+1}_{\langle -m+1 \rangle}(\mathcal{A}[z, z^{-1}], \varepsilon)\big),
    \]
    where \( i_! \) is the functor induced by the inclusion \( \mathcal{A} \hookrightarrow \mathcal{A}[z, z^{-1}] \). For an $R$-module category $\A$, $i_!$ is the familiar functor $R[z, z^{-1}]\otimes_R(-)$.

    By convention, we set
    \[
    L_n^{\langle 1 \rangle}(\mathcal{A}) = L_n(\mathcal{A}) \quad \text{and} \quad L^n_{\langle 1 \rangle}(\mathcal{A}, \varepsilon) = L^n(\mathcal{A}, \varepsilon) \qquad (n \in \mathbb{Z}).
    \]
\end{defn}

As \( L_n \) is four-periodic, the groups \( L_n^{\langle -m \rangle} \) are four-periodic for any fixed \( m \). The groups \( L^n_{\langle -m \rangle} \) are not generally periodic.  The following theorem is an instantiation of the proto-theorem.

\begin{thm}\label{thm:Lmain}
The lower quadratic and symmetric \( L \)-groups satisfy the following identities for all \( m \geq -1 \) and \( n \in \mathbb{Z} \):
\begin{align}
&L_n^{\langle -m \rangle}(\mathcal{A}[z, z^{-1}]) \cong L_n^{\langle -m \rangle}(\mathcal{A}) \oplus L_{n-1}^{\langle -m-1 \rangle}(\mathcal{A}), \label{1} \\
&L^n_{\langle -m \rangle}(\mathcal{A}[z, z^{-1}], \varepsilon) \cong L^n_{\langle -m \rangle}(\mathcal{A}, \varepsilon) \oplus L^{n-1}_{\langle -m-1 \rangle}(\mathcal{A}, \varepsilon), \label{2}\\
%&L_n^{\langle -m \rangle}(\mathcal{C}_1(\mathcal{A})) \cong L_{n-1}^{\langle -m-1 \rangle}(\mathcal{A}), \\
%&L^n_{\langle -m \rangle}(\mathcal{C}_1(\mathcal{A})) \cong L^{n-1}_{\langle -m-1 \rangle}(\mathcal{A}), \\
&L_{m+n+1}^{\langle 1 \rangle}(\mathcal{C}_{m+1}(\mathcal{A}))  \cong L_{m+n}^{\langle 1 \rangle}(\operatorname{Kar}(\mathcal{C}_m(\mathcal{A}))) \cong L_n^{\langle -m \rangle}(\mathcal{A}), \label{3}\\
&L^{m+n+1}_{\langle 1 \rangle}(\mathcal{C}_{m+1}(\mathcal{A}), \varepsilon) \cong L^{m+n}_{\langle 1 \rangle}(\operatorname{Kar}(\mathcal{C}_m(\mathcal{A})), \varepsilon)   \cong L^n_{\langle -m \rangle}(\mathcal{A}, \varepsilon). \label{4}
\end{align}
\end{thm}
\begin{proof}
    Proofs can be found in~\cite{ranicki1992lower}. See Theorem~17.2 for \eqref{1} and \eqref{3}, and Proposition~19.4 for \eqref{2} and \eqref{4}.
\end{proof}
Unlike Theorem~\ref{thm:Main}, the decorations shift along with the degrees. The theorem enables us to make the following definition. 

\begin{defn}
The \emph{ultimate lower quadratic} and \emph{ultimate lower symmetric \( L \)-groups} of an additive category with involution \( \mathcal{A} \) are defined by the direct limits
\[
L_n^{\langle -\infty \rangle}(\mathcal{A}) := \varinjlim_{m \to \infty} L_n^{\langle -m \rangle}(\mathcal{A}), \qquad 
L^n_{\langle -\infty \rangle}(\mathcal{A}, \varepsilon) := \varinjlim_{m \to \infty} L^n_{\langle -m \rangle}(\mathcal{A}, \varepsilon) \qquad (n \in \mathbb{Z}).
\]
The transition maps
\[
L_n^{\langle -m \rangle}(\mathcal{A}) \longrightarrow L_n^{\langle -m-1 \rangle}(\mathcal{A}), \qquad
L^n_{\langle -m \rangle}(\mathcal{A}, \varepsilon) \longrightarrow L^n_{\langle -m-1 \rangle}(\mathcal{A}, \varepsilon)
\]
are given by
$$L_n^{\langle -m \rangle}(\mathcal{A})\cong L_{m+n+1}^{\langle 1 \rangle}(\mathcal{C}_{m+1}(\mathcal{A}))  \rightarrow L_{m+n+1}^{\langle 1 \rangle}(\operatorname{Kar}(\mathcal{C}_{m+1}(\mathcal{A}))) \cong L_n^{\langle -m-1 \rangle}(\mathcal{A}),$$
$$L^n_{\langle -m \rangle}(\mathcal{A}, \varepsilon) \cong L^{m+n+1}_{\langle 1 \rangle}(\mathcal{C}_{m+1}(\mathcal{A}), \varepsilon) \rightarrow L^{m+n+1}_{\langle 1 \rangle}(\operatorname{Kar}(\mathcal{C}_{m+1}(\mathcal{A})), \varepsilon)   \cong L^n_{\langle -m-1 \rangle}(\mathcal{A}, \varepsilon),$$
where the middle arrows are induced by the forgetful map 
$$\mathcal{C}_{m+1}(\mathcal{A})\rightarrow \operatorname{Kar}(\mathcal{C}_{m+1}(\mathcal{A})); A\mapsto (A, \id_A).$$
%(see \cite{ranicki1989additive, ranicki1992lower, ranicki1973algebraicII}).
\end{defn}

Using the ultimate lower \( L \)-groups, we give an instantiation of Theorem~\ref{thm:main2}, which we include for completeness.

Let \( X \) be a pointed subcomplex of the \( n \)-sphere \( S^n \) for some \( n \), and form the open cone \( O(X) \subset \mathbb{R}^{n+1} \). Let \( \mathbb{L}^{\langle -\infty \rangle}(\mathcal{A}) \) and \( \mathbb{L}_{\langle -\infty \rangle}(\mathcal{A}, \varepsilon) \) be the ultimate quadratic and symmetric \( L \)-theory spectra, respectively (see Section 13 of \cite{ranicki1992algebraic} algebraic $L$-spectra). Their homology theories are defined by
\begin{align}
    \mathbb{L}^{\langle -\infty \rangle}(\mathcal{A})_m(X) &:= \lim_{k \to \infty} \pi_{m+k}(\mathbb{L}_k^{\langle -\infty \rangle}(\mathcal{A}) \wedge X),\\
 \mathbb{L}_{\langle -\infty \rangle}(\mathcal{A}, \varepsilon)_m(X) &:= \lim_{k \to \infty} \pi_{m+k}(\mathbb{L}^{k}_{\langle -\infty \rangle}(\mathcal{A}, \varepsilon) \wedge X).
\end{align}

\begin{thm}[Propositions 17.8 and 19.6 in~\cite{ranicki1992lower}]\label{thm:Lmain2}
The \( \mathbb{L}^{\langle -\infty \rangle}(\mathcal{A}) \)-homologies and \( \mathbb{L}_{\langle -\infty \rangle}(\mathcal{A}, \varepsilon) \)-homologies are naturally isomorphic to the corresponding algebraic \( L \)-groups of \( \mathcal{C}_{O(X)}(\mathcal{A}) \), with degree shift:
\begin{align}
    \mathbb{L}^{\langle -\infty \rangle}(\mathcal{A})_{m-1}(X) &\cong L_m^{\langle -\infty \rangle}(\mathcal{C}_{O(X)}(\mathcal{A})),\\
 \mathbb{L}_{\langle -\infty \rangle}(\mathcal{A}, \varepsilon)_{m-1}(X) &\cong L^m_{\langle -\infty \rangle}(\mathcal{C}_{O(X)}(\mathcal{A}, \varepsilon)).
\end{align}

\end{thm}

\begin{cor}
For all \( n \geq 0 \) and \( m \in \mathbb{Z} \),
\[
L_m^{\langle -\infty \rangle}(\mathcal{C}_{n+1}(\mathcal{A})) \cong L_{m-n-1}^{\langle -\infty \rangle}(\mathcal{A}), \quad
L^m_{\langle -\infty \rangle}(\mathcal{C}_{n+1}(\mathcal{A}, \varepsilon)) \cong L^{m-n-1}_{\langle -\infty \rangle}(\mathcal{A}, \varepsilon).
\]
\end{cor}

\begin{proof}
Take \( X = S^n \).
\end{proof}

The following proposition is useful for the classification of Clifford QCAs on \( \Lambda = \mathbb{Z}^m \) with translation symmetry, which are considered in \cite{haah2025topological}.

\begin{prop}[Proposition 17.3 and 19.4  in~\cite{ranicki1992lower}]\label{prop:binomial}
For any \( m \geq 1 \), the quadratic and symmetric \( L \)-groups \( L_* = L_*^{\langle 1 \rangle} \) and \( L^* = L^*_{\langle 1 \rangle} \) of the \( m \)-fold Laurent polynomial extension \( \mathcal{A}[\mathbb{Z}^m] \) of \( \mathcal{A} \) satisfy the binomial formula:
\[
L_n(\mathcal{A}[\mathbb{Z}^m]) = \sum_{i=0}^m \binom{m}{i} L_{n-i}^{\langle 1-i \rangle}(\mathcal{A}), \quad
L^n(\mathcal{A}[\mathbb{Z}^m], \varepsilon) = \sum_{i=0}^m \binom{m}{i} L^{n-i}_{\langle 1-i \rangle}(\mathcal{A}, \varepsilon) \qquad (n \in \mathbb{Z}).
\]
\end{prop}
We end the section with a proposition that offers some relief.
\begin{prop}[Remark below Definition 19.3 in~\cite{ranicki1992lower}]\label{prop:quad2sym}
Quadratic and symmetric $L$-theory are related in the negative degrees:
    $$L^n_{\langle -m\rangle}(\A, \varepsilon) = L_n^{\langle -m\rangle}(\A, \varepsilon) \quad \text{for } n \leq -3 \text{ and } m\geq -1.$$
\end{prop}
For the remainder of this article, we write $L^{*}_{\langle \bullet \rangle}(R,-1):=L^{*}_{\langle \bullet \rangle}(\mathcal A,-1)$, when $\mathcal A$ is the additive category of finitely generated, free $R$-modules. We use the same convention for the quadratic $L$-groups.
\section{Classification of Clifford QCAs}

The Pedersen--Weibel construction is applied to the classification of Clifford QCAs. Firstly, we make the connection from Definition/Observation~\ref{obs:keyobservation} to the \(L\)-groups through \emph{symmetric formations}. Using an $L$-theoretic version of the proto-theorem, we obtain a complete classification of Clifford QCAs in Euclidean spaces. More generally, by invoking an $L$-theoretic analogue of Theorem~\ref{thm:main2}, we classify QCAs defined on open cones over arbitrary subcomplexes of spheres, and relate this classification to the $L$-theory homology theory.  %The extension of these results to general Riemannian manifolds will be addressed in future work.

\subsection{Clifford QCA as formation}
Let $\Lambda$ be a metric space.  Let $\CA$ be the additive category of finitely generated, free modules over the ring $\ZZ_d$ with trivial involution. Recall that the Pedersen--Weibel construction $\mathcal{C}_{\Lambda}(\CA)$ is a filtered additive category with involution. Results below connects Clifford QCA to certain $L$-group of $\mathcal{C}_{\Lambda}(\CA)$. See Section 3 of~\cite{haah2025topological} for a related description in terms of $\varepsilon$-unitary group.  

First, we introduce some notations. Let $(P, \al)$ be a Clifford QCA. For $i\in \Lambda$, $P:=(P(i))_{i\in \Lambda}=(\ZZ_d^{k_i}\oplus \ZZ_d^{k_i})_{i\in\Lambda}=(L(i))_{i\in\Lambda}\oplus (L^*(i))_{i\in\Lambda}=H_{-1}(L)$. One may write $\al=\begin{pmatrix}
        A & B\\
        C & D
    \end{pmatrix}$ with 
    \begin{itemize}
        \item $A\in \Hom_{\mathcal{C}_{\Lambda}(\CA)}(L, L),$
        \item $B\in \Hom_{\mathcal{C}_{\Lambda}(\CA)}(L^*, L),$
        \item $C\in \Hom_{\mathcal{C}_{\Lambda}(\CA)}(L, L^*),$
        \item $D\in \Hom_{\mathcal{C}_{\Lambda}(\CA)}(L^*, L^*).$
    \end{itemize} 
    Moreover, $A, B, C, D$ are subject to the symplectic conditions:
    
$A^{  T}C$ and $B^{  T}D$ are symmetric, and $A^{  T}D - C^{  T}B = I$.

$AB^{  T}$ and $CD^{  T}$ are symmetric, and $AD^{  T} - BC^{  T} = I$.

The usual arithmetic of matrices applies to these blocks as well.

\begin{lem}\label{lem: isom}
    Every Clifford QCA defined according to Definition/Observation~\ref{obs:keyobservation} gives a non-singular $(-1)$-symmetric formation in $\mathcal{C}_{\Lambda}(\CA)$. Two Clifford QCAs giving the same formations are equivalent in $K(\Lambda, \ZZ_d)$.

\end{lem}
\begin{proof}
    Let \(P=(P(i))_{i \in \Lambda}\) and \(\alpha=(\alpha_j^i)_{i,j \in \Lambda}\) be as in Definition/Observation~\ref{obs:keyobservation}. Firstly, $P$ together with the standard symplectic form  on each $P(i)=\ZZ_d^{k_i}\oplus \ZZ_d^{k_i}$ is equal to a hyperbolic $(-1)$-symmetric form  $$(P(i))_{i\in \Lambda}\cong\left(L(i)\oplus L^*(i), \begin{bmatrix}
        0 & 1\\
        -1 & 0
    \end{bmatrix}\right)_{i\in \Lambda}=H_{-1}(L)$$ in $\mathcal C_{\Lambda}(\CA)$ with $L(i)=P^+(i)=\ZZ_d^{k_i}.$
Then $\alpha$ is an automorphism on $$\left(L(i)\oplus L^*(i), \begin{bmatrix}
        0 & 1\\
        -1 & 0
    \end{bmatrix}\right)_{i\in \Lambda},$$ which sends $(L(i))_{i\in \Lambda}$ to some other Lagrangian $\left((\alpha L)(i)=\bigoplus_{j\in \Lambda}\alpha_{i}^j L(j)\right)_{i\in \Lambda}.$ Finally, 
    $$\left(L(i)\oplus L^*(i), \begin{bmatrix}
        0 & 1\\
        -1 & 0
    \end{bmatrix}, L(i), (\alpha L)(i)  \right)_{i\in \Lambda}$$ or simply, $(H_{-1}(L), L, \alpha L)$ is the $(-1)$-symmetric formation given by the Clifford QCA $(P, \alpha)$.

    If $\al, \beta$ give the same formation, $\beta^{-1}\al$ maps $L$ bijectively to itself. In its block form, 
    $\beta^{-1}\al=\begin{pmatrix}
        A & B\\
        0 & (A^T)^{-1}
    \end{pmatrix}$.  Let $\theta=\begin{pmatrix}
        A^{-1} & 0\\
        0 & A^T
    \end{pmatrix}$ which is a separated QCA. Then $ \beta^{-1}\al\theta=\begin{pmatrix}
        1 & BA^T\\
        0 & 1
    \end{pmatrix}$ is equivalent to $\beta^{-1}\al$. For two qudits with Pauli matrices $X_1, X_2, Z_1, Z_2$, the control-$Z$ gate is given by 
    $$X_1\mapsto X_1Z_2$$
    $$X_2 \mapsto Z_1X_2$$
    $$Z_i\mapsto Z_i.$$ For a single qudit with Pauli matrices $X, Z$, the phase gate is given by 
    $$X\mapsto XZ$$
    $$Z\mapsto Z,$$ up to a possible phase. 
    These gates, when applied to various qudits or pairs of qudits, generates any Clifford QCAs of the form 
    $ \begin{pmatrix}
        1 & 0\\
        C & 1
    \end{pmatrix}$ where $C$ is symmetric. Conjugate these using the circuit of Hadamard gates, it is shown that $ \begin{pmatrix}
        1 & BA^T\\
        0 & 1
    \end{pmatrix}$ is also a Clifford circuit. In conclusion, $\al$ and $\beta$ are equivalent. 
    
    %Up to a phase choice, $\beta^{-1}\al\theta$ lifts to a QCA that fixes each Pauli-$Z$ and send each Pauli-$X$ to itself times a product of Pauli-$Z$ matrices.  
 \end{proof}

\begin{lem} \label{lem: septrivial}
    A formation corresponding to a separated QCA is trivial in $L^1(\mathcal{C}_{\Lambda}(\CA), -1)$.
\end{lem}

\begin{proof}
        
    If $\alpha$ is separated, the corresponding formation is isomorphic to
     $$\left(H_{-1}(L), L, L  \right)_{i\in \Lambda},$$ which is trivial in $L^1(\mathcal{C}_{\Lambda}(\CA), -1)$.
\end{proof}

\begin{lem} \label{circuittrivial}
    A formation given by a Clifford circuit is trivial in $L^1(\mathcal{C}_{\Lambda}(\CA), -1)$.
\end{lem}
\begin{proof}
    It suffices to prove the result for a single-layer circuit where there is a partition into uniformly bounded regions $\Lambda=\bigcup_{s}\Lambda_s$ such that
    \begin{itemize}
        \item $\alpha_j^i=0$ whenever, $i\in \Lambda_s$, $j\in \Lambda_{s'}$ for $s\ne s',$
        \item $(\alpha_j^i)_{i, j \in \Lambda_s}: \bigoplus_{k\in \Lambda_s} P(k) \rightarrow \bigoplus_{k\in \Lambda_s} P(k)$ is a symplectic automorphism for each $s\in \Lambda$. 
    \end{itemize} 
    The corresponding formation is a direct sum 
    $$\bigoplus_{s}\left(P_s, \begin{bmatrix}
        0 & 1\\
        -1 & 0
    \end{bmatrix}, L_s, \al L_s\right)=\bigoplus_{s}\left(H_{-1}(L_S), L_s, \al L_s\right),$$ with $P_s:=\bigoplus_{k\in \Lambda_s} P(k)$, $L_s:=\bigoplus_{k\in \Lambda_s} L(k)$.
    This belongs to the image of the map $$\prod_s L^1(\CalC_{\Lambda_s}(\A), -1)$$ given by the inclusions of categories $\prod_s\CalC_{\Lambda_s}(\A)\rightarrow \CalC_{\Lambda}(\A)$. Recall, $\CalC_{\Lambda_s}(\A)\cong \A$ for $\Lambda_s$ is bounded. It suffices to show that $L^1(\A, -1)=L^1(\ZZ_d, -1)=0$. Using the Chinese remainder theorem, we may assume $d$ is a prime power, hence $\ZZ_d$ is a local ring. By Theorem 2 of~\cite{klingenberg1963symplectic}, over a local ring, every symplectic automorphism is elementary\footnote{The reference uses the term transvection.}. See~\cite{ishibashi1999groups} for a further discussion on symplectic automorphism over local ring. This result is sufficient for $L^1(\ZZ_d, -1)=0$ via an equivalent definition of the $L$-group of a ring using symplectic automorphisms. This definition is out of our scope but can be found in textbooks such as section 9.2.3 of~\cite{luck2024surgery}. 
\end{proof}
% \begin{lem}\label{lem: odd_trivial_Zd}
%     The Witt group of formations $L^{1}(\ZZ_d, -1)$ is trivial. 
% \end{lem}
% \begin{proof}
%     By the Chinese Remainder Theorem, it suffices to assume $d=p^r$ for some prime $p.$ Witt group of formation over a field is always zero. Given a  $-1$-symmetric formation $(L\oplus L^*, \begin{bmatrix}
%         0 & 1\\
%         -1 & 0
%     \end{bmatrix}; F, G)$, then $(L/p\oplus L^*/p, \begin{bmatrix}
%         0 & 1\\
%         -1 & 0
%     \end{bmatrix}; F/p, G/p)$ is stably equivalent to an elementary formation. 
% \end{proof}

\begin{lem} \label{lem: surjectivity}
    There is a surjective homomorphism:
    $$K(\Lambda, \ZZ_d)\rightarrow L^1(\mathcal{C}_{\Lambda}(\CA), -1).$$
\end{lem}
\begin{proof}   
    The Lemmas \ref{lem: septrivial} and \ref{circuittrivial} guarantees that the procedure in Lemma \ref{lem: isom} is a well-defined map. Clearly, it respects direct sum. For composition, notice that 
    \begin{align*}
        &(H_{-1}(L), L, \beta \alpha L)\\
        =&(H_{-1}(L), L,  \beta L)+(H_{-1}(L), \beta L, \beta\alpha L)\\
        =&(H_{-1}(L), L,  \beta L)+(H_{-1}(L), L, \alpha L),
    \end{align*}
    where the last equality follows from the formation isomorphism $$\beta: (H_{-1}(L), L, \alpha L)\rightarrow (H_{-1}(L), \beta L, \beta\alpha L).$$
    Surjectivity is implied by Proposition \ref{prop:normalization}. Indeed, for two Lagrangians, $F, G$, Proposition~\ref{prop:normalization} produces a symplectic automorphism $\alpha$ that takes $F$ to $G$. The Clifford QCA $(P, \alpha)$ gives the formation associated to $F, G$. 
\end{proof}
\begin{thm} \label{thm: injectivity}
    The homomorphism in Lemma \ref{lem: surjectivity} is injective. That is, 
    $$K(\Lambda, \ZZ_d)\cong L^1(\mathcal{C}_{\Lambda}(\CA), -1).$$
\end{thm}
\begin{proof}
     The product of Hadamard gate on every $i$ is a quantum circuit whose associated formation is the trivial formation $(H_{-1}(L), L, L^*)$ used in stabilization.

    Let $(P, \alpha)$ be a Clifford QCA. If the formation $(H_{-1}(L), L, \alpha L)=0\in L^1(\mathcal{C}_{\Lambda}(\CA), -1)$, there exists stable isomorphism between $(H_{-1}(L), L, \alpha L)$ and sum of an elementary formation and a trivial formation. Add a circuit of Hadamard gates to $\al$ if needed, we can assume $(H_{-1}(L), L, \alpha L)$ is isomorphic to sum of an elementary formation and a trivial formation. Moreover, conjugation by another QCA does not change its equivalence class, so it might as well be that 
    $$(H_{-1}(L), L, \alpha L)=(H_{-1}(E), E, E^*)\oplus (M, \psi; F, G),$$ with the latter an elementary formation.

By Lemma~\ref{lem: isom}, $\al$ is equivalent to $\beta\oplus \gamma$, where $\beta$ is the circuit of Hadamard gates on $E\oplus E^*$ and $(M, \gamma)$ is a Clifford QCA with $\gamma(F)=G$. As $(M, \psi; F, G)$ is elementary, there exists a Lagrangian $K\subset M$ complementary to both $F$ and $G$. We can identify $F^*$ with $K$ so that the formation becomes $\left(F\oplus K, \begin{pmatrix}
        0 & 1\\
        -1 & 0
    \end{pmatrix}; F, G\right)$ and expand in block form 
$\gamma=\begin{pmatrix}
        A & B\\
        C & D
    \end{pmatrix}$ with 
    \begin{itemize}
        \item $A\in \Hom_{\mathcal{C}_{\Lambda}(\CA)}(F, F),$
        \item $B\in \Hom_{\mathcal{C}_{\Lambda}(\CA)}(K, F),$
        \item $C\in \Hom_{\mathcal{C}_{\Lambda}(\CA)}(F, K),$
        \item $D\in \Hom_{\mathcal{C}_{\Lambda}(\CA)}(K, K).$
    \end{itemize} 
    We claim $A$ is an automorphism of $F$. Since $K$ is complementary to $G$, for any $f\in F$, there exists $g\in G$ and $k\in K$ with $f=g+k$. There is a unique $f_1\in F$ with $g=\gamma(f_1)=Af_1+Cf_1.$ Since $Cf_1+k\in K$ and $M=F\oplus K$, we must have $Cf_1+k=0$ and $Af_1=f$. On the other hand, if $Af=0$ for some $f\in F$, then $\gamma (f)=Cf\in G\cap K=\{0\}.$ For $\gamma$ is an automorphism itself, $f=0$. 

    Next, there is a factorization 
    $$\begin{pmatrix}
        A & B\\
        C & D
    \end{pmatrix}=\begin{pmatrix}
        1 & 0\\
        CA^{-1} & 1
    \end{pmatrix}
    \begin{pmatrix}
        A & B\\
        0 & (A^T)^{-1}
    \end{pmatrix}.$$
\end{proof}
Indeed, let $X=D-CA^{-1}B$, recall that $$A^TD-C^TB=I$$ and $A^TC$ is symmetric,  
\begin{align*}
   A^{T}X&=A^TD-A^TCA^{-1}B \\
   &= I+C^TB-A^TCA^{-1}B\\
   &=I+C^TB-C^TAA^{-1}B\\
   &=I.
\end{align*}
As $A$ is invertible, $X=(A^T)^{-1}$. Repeat the proof of Lemma~\ref{lem: isom} for both factors, we conclude that $\gamma$ is trivial in $K(\Lambda, \ZZ_d)$.

\subsection{Main Result}

Finally, we apply Theorems \ref{thm:Lmain} in order to compute the classification group. 
\begin{cor}\label{cor:Euclid}
    The classification of Clifford QCAs on space $\RR^m$
    \begin{equation*}
        K(\RR^m, \ZZ_d)\cong L^1_{\langle 1 \rangle}(\mathcal C_m(\CA),-1)\cong L^0_{\langle 1 \rangle}(\operatorname{Kar}(\mathcal C_{m-1}(\CA)), -1)\cong L^{1-m}_{\langle 1-m \rangle}(\ZZ_d, -1).
    \end{equation*}
    Furthermore, when $d$ is odd or when $d$ is even and $m\geq 4$, then the classification reduces to quadratic $L$-theory: 
    \begin{equation} \label{symmetric_to_quadratic}
        K(\RR^m, \ZZ_d)\cong L_{1-m}^{\langle 1-m \rangle}(\ZZ_d, -1)=L_{3-m}^{\langle 1-m \rangle}(\ZZ_d).
    \end{equation}
    Lastly, when $d$ is odd and $m\geq 1$ or when $d$ is even and $m\geq 4$, the decoration stabilizes
    $$K(\RR^m, \ZZ_d)\cong L_{3-m}^{\langle -\infty \rangle}(\ZZ_d).$$
    In particular, the classification becomes four-periodic. 
\end{cor}
\begin{proof}
    Combine Theorems~\ref{thm:Lmain} and \ref{thm: injectivity} to get the first chain of isomorphisms. 
    
    To prove~\eqref{symmetric_to_quadratic}, we divide into cases. When $d$ is odd, there is no distinction between quadratic and symmetric forms on free $\ZZ_d$-modules. Indeed, for a finitely generated $\ZZ_d$-modules $M$, the natural map 
    $$(1+T_\varepsilon): Q_\varepsilon(M)=\coker(1-T_\varepsilon)\rightarrow Q^\varepsilon(M)=\ker(1-T_\varepsilon)$$ is an isomorphism with inverse given by $\psi \mapsto \frac{1}{2}\psi+\im(1-T_\varepsilon).$
    
    When $d$ is even, apply Proposition~\ref{prop:quad2sym} to arrive at quadratic $L$-groups for $m\geq 4$.

    Lastly, we rely on the Rothenberg exact sequence~(Theorem 17.2 in~\cite{ranicki1992lower})
    \begin{align*}
        \cdots &\rightarrow \hat{H}^{n+1}(\ZZ_2; K_{-m}(\ZZ_d))\rightarrow L_n^{\langle -m+1\rangle}(\ZZ_d)\rightarrow L_n^{\langle -m\rangle}(\ZZ_d)\\ &\rightarrow \hat{H}^{n}(\ZZ_2; K_{-m}(\ZZ_d))\rightarrow L_{n-1}^{\langle -m+1\rangle}(\ZZ_d)\rightarrow \cdots.
    \end{align*}
Since $\ZZ_d$ is an Artinian ring, $K_{-m}(\ZZ_d)=0$ for $m\geq 1$~(Ex. 4.3 in~\cite{weibel2013k}). The Tate cohomology group $\hat{H}^{n}(\ZZ_2; K_{-m}(\ZZ_d))$ vanishes with its coefficient group. Therefore, 
$$L_n^{\langle -m+1\rangle}(\ZZ_d)\xrightarrow{\sim} L_n^{\langle -m\rangle}(\ZZ_d).$$
   %  By the Chinese Remainder Theorem: 
   % $$L_{n}^{\langle -m \rangle}(\ZZ_d)\cong \bigoplus_{i} L_{n}^{\langle -m \rangle}\left(\ZZ_{p_i^{s_i}}\right),$$ for $d=\prod_i p_i^{s_i}.$ Thus, it is enough to assume $d$ equals a prime power. Any projective $\ZZ_{p^r}$-module is free, therefore, the lower $L$-groups stabilize
   % $$L_{n}^{\langle -m \rangle}(\ZZ_d)=L_{n}^{\langle -\infty \rangle}(\ZZ_d).$$
\end{proof}
\begin{rmk}
Here are some interesting observations:
\begin{itemize}
    \item Previously, it was known that any 1-dimensional QCA is composition of circuits and shifts~\cite{gross2012index}. The latter is separated. In addition, 2-dimensional QCA classification is trivial~\cite{freedman2020classification} (also see the appendix in~\cite{haah2021clifford}). Therefore, \( K(\mathbb{R}^m, \mathbb{Z}_d) = 0 \) for \( m = 0, 1, \) or \( 2 \). The only case in which symmetric (instead of quadratic) \( L \)-groups needs to be calculated is when \( d \) is a power of 2 and \( m = 3 \).
    \item When $d$ is odd, $K(\RR^m, \ZZ_d)$ is four-periodic in $m$. When $d$ is even, $K(\RR^m, \ZZ_d)$ becomes four-periodic for $m\geq 4.$
    \item
    The isomorphism $K(\mathbb R^m,\mathbb Z_d)\cong L^0(\operatorname{Kar}(\mathcal C_{m-1}(\mathcal A)),-1)$ can be given an interpretation based on \emph{invertible subalgeras} whose definition can be found in~\cite{haah2023invertible}. In particular, it is consistent with the conclusion that the classification of $m$-dimensional QCA coincides with the classification of $(m-1)$-dimensional invertible subalgebras. In Section~4 of~\cite{haah2023invertible}, the invertible subalgebra associated to a translation-invariant Clifford QCA is shown to be equivalent to a non-singular skew-symmetric form. We propose that, in the absence of translation symmetry, a Clifford invertible subalgebra is a non-singular skew-symmetric forms on an object in $\operatorname{Kar}(\mathcal C_{m-1}(\mathcal A))$. Their classification is precisely given by the Witt group $L^0_{\langle1\rangle}(\operatorname{Kar}(\mathcal C_{m-1}(\mathcal A)),-1)\cong K(\R^m, \ZZ_d)$.
\end{itemize}

   \end{rmk}

We apply Theorem~\ref{thm:Lmain2} to further the conclusion.

\begin{cor}
    Let $O(X)$ be an open cone of $X$ a pointed subcomplex of $S^n$ and $\Lambda=O(X)\times \RR^m$ for $n, m\geq 0$. The nontrivial Clifford QCAs on $\Lambda$ are classified by generalized homology groups 
    $$K(\Lambda, \ZZ_d)\cong  \mathbb L^{\langle -\infty \rangle}(\ZZ_d)_{2-m}(X),$$ for any $m$ greater than or equal to $m_d =\begin{cases}
        0, \text{for $d$ odd},\\
        4, \text{for $d$ even}.
    \end{cases}$ 
\end{cor}
\begin{proof}
Apply Theorem \ref{thm:Lmain2},
    \begin{align*}
        K(\Lambda, \ZZ_d)&\cong L^1(\mathcal C_\Lambda(\CA), -1)\\
        & \cong L^1(\mathcal C_m(\mathcal C_{O(X)}(\CA)), -1)\\
        &\cong L^{1-m}_{\langle -m \rangle}(\mathcal C_{O(X)}(\CA),-1)\\
        %&\cong L^{1-m}_{\langle -m \rangle}(\mathcal C_{O(X)}(\CA),-1)\\
        &\cong L_{3-m}^{\langle -\infty \rangle}(\mathcal C_{O(X)}(\CA))\\
        &\cong \mathbb L^{\langle -\infty \rangle}(\ZZ_d)_{2-m}(X).
    \end{align*}
\end{proof}
An example is that $K(\Lambda, \ZZ_d)=0$ for any open cone over a contractible complex $X$. This applies in particular to half-Euclidean spaces, which are special cases of such open cones.
\subsection{Homotopy and Asymptotic Geometry in Lattice Models}

In the study of lattice models, it is natural to go beyond systems defined on full Euclidean space and consider models supported on more general subspaces. A particularly broad and physically relevant class consists of spaces that are coarsely equivalent to open cones. This class includes half-Euclidean spaces, cones on a plane, and many other geometries arising from boundaries, defects, or engineered lattice terminations. Even in one dimension, one may replace the line by a tree with more than two ends, already leading to qualitatively new large-scale behavior. As the dimension increases, the family of such spaces becomes increasingly rich.

Our main conceptual contribution is to show that the classification of Clifford quantum cellular automata (QCA) naturally takes the form of a generalized homology theory of the boundary at infinity of the underlying lattice. This is striking in two aspects. First, the classification depends only on the asymptotic geometry of the lattice and is insensitive to microscopic details. Consequently, two lattices that differ substantially at short scales—such as by the introduction of random local defects—can nevertheless share the same Clifford QCA classification.

Second, for lattices coarsely equivalent to open cones, the resulting invariant is a homotopy invariant of the boundary at infinity. In particular, the vanishing of the Clifford QCA classification on a half-space extends to cones over any contractible complex, including cones with different opening angles or local geometric distortions. From this perspective, the classification is governed not by the local geometry of the lattice, but by the homotopy type of its asymptotic boundary.

We emphasize that this phenomenon is not unique to Clifford QCA. A closely related example appears in recent work of ours~\cite{artymowicz2024mathematical}, which develops a mathematical framework showing that the Hall conductance—and its non-abelian and higher-dimensional generalizations for gapped lattice systems—takes values in the \v{C}ech cohomology of the boundary at infinity of the system. The close parallel between these results suggests a broader principle: topological invariants of quantum lattice phases are naturally encoded in (co)homological data associated with the geometry at infinity. We expect this perspective to be useful in unifying and extending classification results across a wide range of lattice systems.

\section{Discussion}
Inspired by a construction of Pedersen and Weibel, we relate Clifford QCAs on metric spaces to algebraic \( L \)-theory and the associated generalized homology theory. In this work, we have focused on Clifford QCAs without symmetries. For Clifford QCAs equipped with symmetries—either spatial or internal—given by the action of a group \( G \), versions of Theorem~\ref{thm: injectivity} would exist. The classification boils down to $L$-groups of a suitable category $\CA[G]$ of \( \ZZ_d[G] \)-modules. Note that $\CA[G]$ may not be a category of free modules, depending on how \( G \) acts. The Pedersen--Weibel construction can be applied to classify these $G$-equivariant Clifford QCA on a space. When $G$ acts by spatial translation, we may interpret $\CA[G]$ itself as already representing a system of qudits defined on a Cayley graph of $G$. It is natural to question the relation to Clifford QCA on this graph without symmetry pertaining to the category $\mathcal C_G(\CA)$.\footnote{We choose this notation because this Pedersen--Weibel category does not depend on choice of generating set, for it does not affect the large scale geometry of Cayley graph.} We offer a limited answer.

Define a process, called \emph{symmetry coarse-graining}, in which the symmetry is allowed to break arbitrarily, provided that some finite-index subgroup of symmetries is preserved. Mathematically, it corresponds to the direct limit
\[
\varinjlim_{H \leq G \text{ with } [G : H] < \infty}  L^1\big(\mathcal{A}[H], -1\big).
\]
For \( G = \ZZ^n \), with \( \mathcal{A}[\ZZ^n] \) the category of finitely generated, free \( \ZZ_d[\ZZ^n] \)-modules, then by Proposition \ref{prop:binomial}
\[
L^1(\mathcal{A}[\mathbb{Z}^m], -1) = \sum_{i=0}^m \binom{m}{i} L^{1-i}_{\langle 1-i \rangle}(\mathcal{A},-1).
\]
All finite-index subgroups are isomorphic to $\ZZ^m$. It turns out the only the last term survives coarse-graining~\cite{haah2025topological}. That is

\[
\varinjlim_{\ZZ^m \leq \ZZ^m} L^1\big(\mathcal{A}[\ZZ^m], -1\big)\cong L^{1-m}_{\langle 1-m\rangle}(\ZZ_d,-1).
\] Interestingly, this agrees with Corollary~\ref{cor:Euclid} where no translation-symmetry is assumed. However, we emphasize that this agreement does not generally hold when \( G \) is nonabelian~\cite{weinberger-pc}. Hence, the relationship between $$
\varinjlim_{H \leq G \text{ with } [G : H] < \infty}  L^m\big(\mathcal{A}[H], -1\big)$$ and $L^m(\mathcal{C}_G(\CA), -1)$ is still a mystery. 

Another natural generalization is to consider Clifford QCAs on lattices with mixed qudit dimensions. By the Chinese Remainder Theorem, the problem separates into cases with qudits of dimensions \( \{p, p^2, \dots, p^s\} \) for some prime \( p \) and natural number \( s \). This amounts to choosing \( \mathcal{A} \) as the category of \( \ZZ_{p^s} \)-modules of the form
$$\bigoplus_{r=1}^s \ZZ_{p^r}^{k_r}.$$ Actually, this is the category of all finitely generated \( \ZZ_{p^s} \)-modules.
Similarly, for the setup with $G$-symmetry, we can select $\CA[G]$ to be the $\ZZ[G]$-modules of the form $$\bigoplus_{r=1}^s \ZZ_{p^r}^{k_r}[G].$$ For $G=\ZZ^m$, this is exactly what was studied in~\cite{ruba2024homological} under the name \emph{quasi-free modules}.

Lastly, in view of the growing popularity of stabilizer codes and their boundaries \cite{liang2024operator, liang2025generalized, liang2025planar, ruba2025witt}, we envision a version of the Witt group relevant to this setting. The bulk-boundary correspondence would be explained through yet another instantiation of the proto-theorem. Foundational work in this direction is being laid, building on the notion of quasi-symplectic modules introduced in \cite{ruba2024homological,ruba2025witt}.
\appendix
\section{Quantum Spin System and QCA}
\textit{Quantum spin systems} are a fundamental class of models in condensed matter physics that describe the collective behavior of interacting spins in a lattice. These systems are essential for understanding a wide range of physical phenomena, including magnetism, phase transitions, and quantum entanglement. The study of quantum spin systems has led to insights into the nature of quantum many-body systems and has applications in quantum information science, statistical mechanics, and topological phases of matter. 

A quantum spin system consists of a set $\Lambda$ with a metric $\rho$, often called a lattice\footnote{Following the custom of condensed matter physics, the term ``lattice'' does not indicate group structure.}, as well as an algebra of local observables $\SA_{\mathrm {loc}}$.  Let $\mathcal P_0(\Lambda)$ be the set of finite subsets of $\Lambda.$  We denote $\bigotimes_{x\in X}\mathrm{Mat}(\CC^{d_x})$ by $\SA_X$ for $X\in \mathcal{P}_0(\Lambda)$, where $\mathrm{Mat}(\CC^{d_x})$ is the matrix algebra on $\CC^{d_x}$. For $X, Y \in \mathcal{P}_0(\Lambda)$ with $Y\subset X$, define $\iota_{Y,X}: \SA_Y\rightarrow \SA_X$ by $\iota_{Y,X}(A)= A\otimes \Id_{X\backslash Y}$. Then $\{\SA_X\}_{X\in P_0(\Lambda)}$ is a directed system. We say $A\in \SA_X$ is supported on $Y \subset X$ if $A=A'\otimes \Id_{X\backslash Y}$ for some $A'\in \SA_Y$. Here $\Id_{X\backslash Y}$ is the identity operator in $\SA_{X\backslash Y}.$ The \textit{support} of $A$ is the intersection of such subsets. 

The algebra of \textit{local observables} is defined as
\begin{equation}
    \SA_{\mathrm {loc}}= \varinjlim_{X\in \mathcal P_0(\Lambda)}\SA_X,
\end{equation}
  For finite $\Lambda$, instead, one may define $\SA_{\mathrm {loc}}=\mathrm{Mat}(\mathcal H)$ with 
$$\mathcal{H}=\bigotimes_{x\in \Lambda} \CC^{d_x}.$$ As $\mathrm{Mat}(\CC^{d_x}\otimes \CC^{d_y})\cong \mathrm{Mat}(\CC^{d_x})\otimes\mathrm{Mat}(\CC^{d_y})$, the two definitions are equivalent. 
\begin{prop}
    $\SA_{\mathrm {loc}}$ is a $*$-algebra with the operator norm. Moreover, its norm completion $\SA=\overline{\SA_{\mathrm {loc}}}$ is a $C^*$-algebra. 
\end{prop}
\begin{proof}
See section 6.2 of \cite{bratteli2012operator}.    
\end{proof}
\begin{rmk}
The metric $\rho$ does not matter for the abstract $C^*$-algebraic structure of $\SA$. Its role becomes explicit when we consider the local automorphisms of $\SA$.  
\end{rmk}

Fix a uniform $d_x=d$ for all $x\in \Lambda$. This allows us to canonically identify all $\SA_{\{i\}}$, for $i\in \Lambda.$ The main object of interest is a group of $*$-automorphisms of $\SA$ with bounded propagation.
\begin{defn}
    A \textit{quantum cellular automaton (QCA)} $\alpha$ is a $*$-automorphism of $\SA$ with a constant called its \textit{range} $r\geq0$ such that 
    $$\alpha(A)\in \SA_{B_r(i)}$$ for any $A\in \SA_{\{i\}}$ and $i\in \Lambda$. Here $B_r(i):=\{j\in \Lambda : \rho(i,j)\leq r\}$.
\end{defn}

\begin{prop}
Let $\alpha$ be a QCA with range $r$.
    \begin{enumerate}
        \item $\alpha$ is a self-isometry of the $C^*$-algebra $\SA$. 
        \item For any $A\in \SA_X$ with $X\subset\Lambda$ finite, $\alpha(A)\in \SA_{B_{r}(X)}$ where $B_{r}(X):=\{j\in \Lambda : \rho(X,j)\leq r\}$
        \item Its inverse $\alpha^{-1}$ is a QCA with the same range $r$. 
        \item The set $\mathbf Q(\SA)$ of all QCAs is a group.
    \end{enumerate}
\end{prop}

\begin{proof}
 \begin{enumerate}
     \item A $C^*$-isomorphism is isometric. 
     \item Write $A=\sum_{k}A_k$ with $A_k=\otimes_{i\in X}A_k^{(i)}$ and $A_k^{(i)}\in\SA_{\{i\}}$. Note each $\alpha(A_k^{(i)})\in \SA_{B_{r}(X)}.$
     \item For any $A\in \SA_{\{x\}}$ and $B\in \SA_{\{y\}}$, $[\alpha^{-1}(A), B]=[A, \alpha(B)]=0$ for $\rho(x,y)>r.$ In other words, $\alpha^{-1}(A)$ is in the commutant of $\SA_{\Lambda \backslash B_r(x)}$. Therefore, it is supported in $B_r(x).$
     \item From 2, composition of two QCAs is a QCA. From 3, the inverse of a QCA is a QCA. 
 \end{enumerate}  
\end{proof}

\begin{exmp}
    Let $\Lambda = \ZZ$, and suppose that the algebras $\SA_{\{i\}}$ are all mutually isomorphic via a fixed family of isomorphisms
\[
t_i \colon \SA_{\{i\}} \to \SA_{\{i+1\}} .
\]
Right translation then defines a QCA
\[
\tau \colon A \in \SA_{\{i\}} \longmapsto t_i(A) \in \SA_{\{i+1\}},
\]
for all $i \in \ZZ$.

\end{exmp}

\begin{exmp}
    For general $\Lambda$ with metric $\rho$, a translation is a bijective map $T: \Lambda\rightarrow \Lambda$ such that $\rho(T(i), i)<r$ where $r$ depends only on $T$. Define a QCA by permuting $\SA_i$ according to $T$. 
\end{exmp}

\begin{exmp}
    Another important class of examples is conjugation by quantum circuits as defined below.  
\end{exmp}

\begin{defn}
    A \textit{quantum gate} $G$ is a local unitary operator, i.e. $G\in \mathbf U(\SA_X) :=\{A\in \SA_X: A^*A=AA^*=I\}$. Recall that the support of $G$ is contained in $X$.\\
    A \textit{single-layer quantum circuit} is a formal product $$\prod_{i\in I}G_i,$$ where $\{G_i\}_{i\in I}$ are quantum gates whose supports $\{X_i\}_{i\in I}$  are pairwise disjoint and uniformly bounded. Though not an observable itself, it defines a QCA through conjugation. In fact, we often abuse the term single-layer quantum circuits to mean this QCA. \textit{A (finite-depth) quantum circuit}  $\mathscr C$ is a composition of finitely many single-layer circuits.
\end{defn}

\begin{prop}[Lemma 2.10 in~\cite{freedman2020classification}]
    Conjugation by a quantum circuits (or simply quantum circuits) form a normal subgroup $\mathbf{C}(\SA)\triangleleft\mathbf{Q}(\SA)$.
\end{prop}

\begin{defn}
    Given two algebras $$\SA_{\mathrm{loc}}=\varinjlim_{\Gamma\in \mathcal P_0(\Lambda)}\bigotimes_{x\in \Gamma}\mathrm{Mat}(\CC^d)$$ and $$\SA'_{\mathrm{loc}}=\varinjlim_{\Gamma\in \mathcal P_0(\Lambda)}\bigotimes_{x\in \Gamma}\mathrm{Mat}(\CC^{d'})$$ defined over the same metric space $(\Lambda, \rho)$, their tensor product $\SA_{\mathrm{loc}}\otimes \SA'_{\mathrm{loc}}$ is defined point-wise as
    $$\varinjlim_{\Gamma\in \mathcal P_0(\Lambda)}\bigotimes_{x\in \Gamma}\left(\mathrm{Mat}(\CC^d)\otimes \mathrm{Mat}(\CC^{d'})\right).$$
\end{defn}

\begin{defn}
    Observe that $\mathbf Q(\SA)\hookrightarrow \mathbf Q(\SA\otimes \SA')$ with $\alpha \mapsto \alpha \otimes \id_{\SA'}$. This is known as \textit{stabilization}. Let $\mathbf Q(\Lambda)$ be the union of the resulting sequence
    $$\mathbf Q(\SA)\hookrightarrow \mathbf Q(\SA^{\otimes 2})\hookrightarrow\cdots \hookrightarrow \mathbf Q(\SA^{\otimes k})\hookrightarrow\dots.$$ Define $\mathbf C(\Lambda)$ similarly.
\end{defn}
The following fundamental theorem of QCA is reminiscent of the algebraic $K_1$ group. 
\begin{theorem}[Theorem 2.3 in~\cite{freedman2020classification}]
     Using notations above, $\mathbf C(\Lambda)$ is a normal subgroup of $\mathbf Q(\Lambda)$. Moreover, the quotient group $\overline{\mathbf Q}(\Lambda)=\mathbf Q(\Lambda)/\mathbf C(\Lambda)$ is abelian. 
\end{theorem}

The generalized Pauli matrices, often referred to as \textit{Clock and Shift matrices}, are a family of operators that extend the well-known Pauli matrices to higher-dimensional spaces. These matrices play a central role in quantum mechanics, quantum information theory, and the study of finite-dimensional operator algebras. For a \(d\)-dimensional Hilbert space \(\mathcal{H}\) or a \emph{qudit}, the generalized Pauli matrices are defined by two fundamental operators: the \textit{shift operator} \(X\) and the \textit{clock operator} \(Z\). These operators act on a fixed set of basis \(\{|0\rangle, |1\rangle, \dots, |d-1\rangle\}\) of \(\mathcal{H}\) as follows: the shift operator \(X\) performs a cyclic permutation \(X|k\rangle = |k+1 \mod d\rangle\), while the clock operator \(Z\) introduces a phase factor \(Z|k\rangle = \xi^k |k\rangle\), where \(\xi = e^{2\pi i/d}\) is a primitive \(d\)-th root of unity. These operators satisfy the  commutation relation \(XZ = \xi^{-1} ZX\), which generalizes the anti-commutation relation of the standard Pauli matrices. Importantly, the set of all products of the form \(X^a Z^b\) for \(a, b \in \{0, 1, \dots, d-1\}\) forms a basis for the space of \(d \times d\) complex matrices \(\mathrm{Mat}(\mathbb{C^d})\).  

\begin{defn}
    Let $\Lambda$ be a metric space with metric $\rho$. On each site $i\in \Lambda$, we assign $\mathcal H^{\otimes n_i}$. Therefore, 
    $\SA$ is the norm completion of 
    $$\varinjlim_{\Gamma\in \mathcal P_0(\Lambda)}\bigotimes_{i\in \Gamma}\left(\mathrm{Mat}(\CC^d)\right)^{\otimes n_i},$$ where $n_i$ may vary from site to site. Denote by $X_{k,i}, Z_{k,i}$ the generalized Pauli matrices over the $k$-th qudit on $i\in \Lambda$. Define the \textit{generalized Pauli operators} $\mathbf P(\SA)\subset \mathbf U(\SA)=\{A\in \SA: A^*A=AA^*=I\}$ as the group generated by $\{X_{k,i}, Z_{k,i}\}_{k,i}$ and $\{\lambda I: \lambda\in U(1)\}$. 
\end{defn}
For now, we omit the index $k$ and assume there is one qudit per $i \in \Lambda$. Everything is adaptable to the case with multiple qudits per site. 
\begin{prop}
    A QCA $\alpha$ is completely determined by the images $\alpha(X_i), \alpha(Z_i)\in \mathbf U(\SA)$ for every $i\in \Lambda$ and $X_i, Z_i\in \SA_{\{i\}}$ the clock and shift matrices. In other words, it suffices to look at the injective group homomorphism $\bar\alpha: \mathbf P(\SA) \rightarrow \mathbf U(\SA)$.
\end{prop}
\begin{proof}
    Since $\mathbf P(\SA)$ generates $\SA$, the proposition follows from the fact $\alpha$ is an automorphism. 
\end{proof}

\begin{defn}
     A QCA $\alpha$ is called \textbf{separated} if  $\alpha(\langle X_i\rangle_{i\in \Lambda})=\langle X_i\rangle_{i\in \Lambda}$ and $\alpha(\langle Z_i\rangle_{i\in \Lambda})=\langle Z_i\rangle_{i\in \Lambda}$. Here, $\langle X_i\rangle_{i\in \Lambda}$ (resp. $\langle Z_i\rangle_{i\in \Lambda}$) denotes the subgroup generated by $\{X_i, \lambda I: i\in \Lambda, \lambda \in U(1)\}$ (resp. $\{Z_i, \lambda I: i\in \Lambda, \lambda \in U(1)\}$).
\end{defn}

Translations are clearly separated as they act as mere permutations on $\{X_i\}_{i\in \Lambda}$ and $\{Z_i\}_{i\in \Lambda}$, respectively. 

\begin{defn} \label{defn: CQ}
    A Clifford QCA is a group automorphism $\alpha$ of $\mathbf P(\SA)$ with bounded propagation, i.e., there exists $r>0$ such that, for every $i\in \Lambda$, $\alpha(X_i)$ and $\alpha(Z_i)$ are, respectively, products of elements in $\{X_j, Z_j: \rho(i,j)<r\}$. We denote the group of Clifford QCAs by $\mathbf{CQ}(\SA)$. 
\end{defn}
\begin{defn}
    A quantum gate $G\in \mathbf U(\SA_X)$ is \textit{Clifford} if $G\mathbf P(\SA)G^{-1}=\mathbf P(\SA)$. A \textit{Clifford circuit} is a circuit consisting of Clifford gates. We denote by $\mathbf{CC}(\SA)$ the collection of QCA's equal to conjugation by Clifford circuits.
\end{defn}
\begin{exmp}
    The product $\mathcal X=\prod_{i\in \Lambda}X_i$ is a Clifford circuit. Indeed, $\Ad_\mathcal{X} X_i=X_i$ and $\Ad_\mathcal{X} Z_i=\xi Z_i.$ 
\end{exmp}

The quotient of $\mathbf P(\SA)$ by its center $Z(\mathbf P(\SA))$ is an abelian group. We describe a convenient representation of $P(\Lambda, d):=\mathbf P(\SA)/Z(\mathbf P(\SA))$. Let $\Lambda$ be a lattice. Then $P(\Lambda, d)$ splits into a direct sum $$P(\Lambda, d)=\bigoplus_{i\in \Lambda} P(i)=\bigoplus_{i\in \Lambda} \ZZ_d^2.$$ Over each summand, the clock and shift operators $X$ and $Z$ on the corresponding site $i\in \Lambda$ are represented respectively by $\colvec{1}{0}$ and $\colvec{0}{1}$ in $\ZZ_d^2$. Products of $X$ and $Z$ are represented by a sum of $\colvec{1}{0}$ and $\colvec{0}{1}$, which is only unique up to a phase. For example, both $X^2Z$ and $XZX$ are represented by $\colvec{2}{1}$. As $X^d= Z^d= I$, we have $\colvec{d}{0}=\colvec{0}{d}=\colvec{0}{0}$ which is true in $\ZZ_d^2.$

Morevover, commutation relation between two operators is recovered by the matrix $\begin{pmatrix}
    0 & 1\\
    -1 & 0
\end{pmatrix}$ as follows. Given two Pauli operators such as $X^aZ^b, X^eZ^f$ represented by $\colvec{a}{b}$ and $\colvec{e}{f}$, they commute up to a phase
$$\xi^{af-be}X^aZ^b X^eZ^f=X^eZ^fX^aZ^b.$$
 The phase is given by the standard symplectic form 
\begin{equation}
\omega: \ZZ_d^2\times \ZZ_d^2 \rightarrow \ZZ_d, 
\end{equation}$$\omega\left(\colvec{a}{b}, \colvec{e}{f}\right)=
\begin{pmatrix}
    a & b
\end{pmatrix} 
\begin{pmatrix}
    0 & 1\\
    -1 & 0
\end{pmatrix}  \colvec{e}{f}=af-be.
$$
The sum $\Omega:=\bigoplus_{i\in \Lambda}\omega$ is the standard symplectic form on $P(\Lambda, d)$.

For brevity, we write $P$ for the abelian group $P(\Lambda,d)$ equipped with the standard symplectic form $\Omega$. 
\begin{defn} \label{defn: local symp}
A local symplectic automorphisms of $P$ consists of the following data:
\begin{itemize} 

\item A collection of abelian group homomorphisms $\al=(\alpha_j^i: P(i)\rightarrow P(j))_{i, j\in \Lambda}$ that are
    \begin{enumerate}
        \item symplectic: $\alpha^i_j$ pulls back the standard symplectic form, and
        \item invertible: there exists $(\beta_j^i: P(i)\rightarrow P(j))_{i, j\in \Lambda}$ such that
    $$(\beta \alpha)^i_l = \sum_{j \in \Lambda} \beta^j_l \alpha^i_j=\delta_{il}.$$
    
\end{enumerate}

\item A constant $r>0$ depending only on $\alpha$ such that $\alpha^i_j=0$ whenever $\rho(i,j)>r$.
\end{itemize}
\end{defn}
\begin{lem} \label{lem: motivation}
    There is a surjective homomorphism $$\kappa: \mathbf{CQ}(\SA)\rightarrow \mathrm{Aut}_{\mathrm{loc}}^{\Omega}(P),$$ where $\mathrm{Aut}_{\mathrm{loc}}^{\Omega}(P)$ is the group of local symplectic automorphisms of $P$ defined in Definition~\ref{defn: local symp}. Furthermore, $\ker \kappa\subset \mathbf{CC}(\SA)$. 
\end{lem}
\begin{proof}
An Clifford QCA $\alpha$ preserves center $Z(\mathbf P(\SA))$. Therefore, it induces an abelian group automorphism $\kappa(\alpha)$ on $\mathbf P(\SA)/Z(\mathbf P(\SA))=P$. Given $p, p'\in P$, one could lift them uniquely up to phases to $\tilde p, \tilde p' \in \mathbf P(\SA)$. Then $m=\Omega(\kappa(\alpha) p, \kappa(\alpha) p')=\Omega( p, p')$ because  
$$\xi^m =\al (\tilde p')\al (\tilde p)(\al (\tilde p)\al (\tilde p') )^{-1}=\tilde p' \tilde p (\tilde p \tilde p')^{-1}.$$ 
Here, the RHS does not depend on the choice of lifting.

The bounded propagation condition in Definition~\ref{defn: CQ} translates to the existence of constant $r$ in Definition~\ref{defn: local symp}

If $\alpha\in \ker \kappa$, then $\alpha(X_i)=\xi^{m_i}X_i$ and $\alpha(Z_i)=\xi^{n_i}Z_i$. There is a single-layer quantum circuit with gates $X_i^{n_i}Z_i^{-m_i}$ such that $\Ad_{X_i^{n_i}Z_i^{-m_i}}(X)=\xi^{m_i}X_i$ and $\Ad_{X_i^{n_i}Z_i^{-m_i}}(Z)=\xi^{n_i}Z_i.$ In other words, $\ker \kappa\subset \mathbf{CC}(\SA)$. 

To prove surjectivity, let $\theta \in \mathrm{Aut}_{\mathrm{loc}}^{\Omega}(P)$. Define $\tilde \theta (X_i)\in \mathbf P(\SA)$ (resp. $ \tilde \theta (Z_i) \in \mathbf P(\SA)$) to be a choice of lifting of the image under $\theta$ of $\colvec{1}{0}$ $\left( \text{resp. } \colvec{0}{1}\right)$ in $P(i)$ for each $i\in \Lambda$. These choices are independent of each other. Then extend $\tilde \theta$ to $\mathbf P(\SA)$ by the axioms of an automorphism. Bounded propagation condition can be seen from the corresponding locality condition of $\theta,$ therefore, $\tilde \theta\in \mathbf{CQ}(\SA)$. It is not hard to see that $\kappa(\tilde \theta)=\theta$.
\end{proof}

For classification purposes, quantum circuits are considered trivial. Consequently, it suffices to start our analysis from the target $\mathrm{Aut}_{\mathrm{loc}}^{\Omega}(P)$. Observation/Definition \ref{obs:keyobservation} is based on this result.

 \section{Proof of Theorem~\ref{thm:Main}}
We include a detailed proof of the first isomorphism in Theorem~\ref{thm:Main}. All steps are included, both for completeness and due to its resemblance to a pumping argument familiar from physics.

Let $[A, \alpha]$ be a class in $K_1(\mathcal{C}_{i+1}(\CA))$, where $A\in \mathcal{C}_{i+1}(\CA)$ and $\alpha: A\rightarrow A$ an automorphism. Recall $A$ has components $A(j_1,\dots, j_{i+1}) \in \CA$ for $j_1, \dots, j_{i+1}\in \ZZ.$ Decompose $A=A^-\oplus A^+$ with 
$$A^-(j_1, \dots, j_{i+1}):=\begin{cases}
    0, &\text{ if } j_{i+1}\geq 0\\
    A(j_1, \dots, j_{i+1}), &\text{ if } j_{i+1}< 0
\end{cases},$$
$$A^+(j_1, \dots, j_{i+1}):=\begin{cases}
    A(j_1, \dots, j_{i+1}), &\text{ if } j_{i+1}\geq 0\\
    0, &\text{ if } j_{i+1}< 0
\end{cases}.$$
We denote the projection to $A^-$ by $p^A_-$. Consider $\alpha  p^A_-  \alpha^{-1}$. It is a projection because 
$$\alpha  p^A_-  \alpha^{-1}  \alpha  p^A_-  \alpha^{-1}=\alpha  p^A_-  \alpha^{-1}.$$ Let $r$ be the filtration degree of $\alpha$, $\alpha  p^A_-  \alpha^{-1}$ is the identity on $A(j_1, \dots, j_{i+1})$ if $j_{i+1}< -2r$, and the $0$-map if $j_{i+1}>2r$.

Let $\bar A(j_1, \dots, j_i):=\bigoplus_{j=-2r}^{2r}A(j_1,\dots, j_i, j)\in \mathcal{C}_{i}(\CA).$ Define $\phi([A, \alpha])=[\bar A, \alpha  p^A_-  \alpha^{-1}]-[\bar A,  p^A_- ]$ in $K_0(\operatorname{Kar}(\mathcal{C}_{i}(\CA)))$. 

\begin{lem}\label{lem:elementary}

Let $A$ and $B$ be objects of $\mathcal{C}_{i+1}(\CA)$ and $\psi : A \oplus B \to A \oplus B$ a bounded projection satisfying
\[
\psi|_{(A \oplus B)(j_1, \dots, j_{i+1})} = 
\begin{cases}
0 & \text{if } j_{i+1} > k \\
1 & \text{if } j_{i+1} < -k
\end{cases}
\]
for some $k$. Let $\gamma : A \oplus B \to A \oplus B$ be an elementary isomorphism\footnote{Recall an elementary isomorphism has matrix
\[
 \begin{pmatrix} 1 & \eta \\ 0 & 1 \end{pmatrix}, \quad \eta : B \to A.
\]} with matrix
\[
\gamma = \begin{pmatrix} 1 & \eta \\ 0 & 1 \end{pmatrix}, \quad \eta : B \to A.
\]
Then $\psi$ and $\gamma \psi \gamma^{-1}$ restricted to a sufficiently big band around $j_{i+1} = 0$ represent the same element of $K_0(\operatorname{Kar}(\mathcal{C}_{i}(\CA)))$.
\end{lem}

\begin{proof}
 
Assume $\eta$ has filtration degree $l > k$. Choose $B'$ and $B''$ such that
\[
B'(j_1, \dots, j_{i+1}) =
\begin{cases}
B(j_1, \dots, j_{i+1}) & \text{if } |j_{i+1}| \leq 2l \\
0 & \text{if } |j_{i+1}| > 2l
\end{cases}
\]
and $B = B' \oplus B''$. Also define $\eta', \eta'' : B \to A$ as the composites
\[
B \to B' \oplus 0 \to B \xrightarrow{\eta} A \quad \text{and} \quad B \to 0 \oplus B'' \to B \xrightarrow{\eta} A.
\]
Writing
\[
\gamma' = \begin{pmatrix} 1 & \eta' \\ 0 & 1 \end{pmatrix}, \quad \gamma'' = \begin{pmatrix} 1 & \eta'' \\ 0 & 1 \end{pmatrix},
\]
it is clear that $\gamma = \gamma'' \cdot \gamma' = \gamma' \cdot \gamma''$.

But
\[
\gamma \psi \gamma^{-1} = \gamma' \gamma'' \psi (\gamma'')^{-1} (\gamma')^{-1} = \gamma' \psi (\gamma')^{-1},
\]
 since $\psi$ is $1$ or $0$ outside a small band around $j_{i+1} = 0$ and $\gamma''$ is the identity in a bigger band around $j_{i+1} = 0$. As $\gamma'$ restricts to an isomorphism in the band, $\gamma' \psi (\gamma')^{-1}$ and $\psi$ are equivalent in $K_0(\operatorname{Kar}(\mathcal{C}_{i}(\CA))).$
   
\end{proof}

\begin{lem}
    The construction $\phi: K_1(\mathcal{C}_{i+1}(\CA))\rightarrow K_0(\operatorname{Kar}(\mathcal{C}_{i}(\CA)))$ is a well-defined homomorphism. 
\end{lem}
\begin{proof}
    Given a diagram in $\mathcal{C}_{i+1}(\CA)$
        \[
\begin{array}{ccccccccc}
0 & \longrightarrow & A & \longrightarrow & A\oplus B & \longrightarrow & B & \longrightarrow & 0 \\
  &                 & \downarrow{\scriptstyle \alpha} 
  &                 & \downarrow{\scriptstyle \gamma} 
  &                 & \downarrow{\scriptstyle \beta} 
  &                 &   \\
0 & \longrightarrow & A & \longrightarrow & A \oplus B & \longrightarrow & B & \longrightarrow & 0
\end{array},
\]
then $\gamma  (\alpha^{-1}\oplus \beta^{-1})$ is an elementary isomorphism. Apply Lemma \ref{lem:elementary}, 
\begin{align}
    \phi([A\oplus B, \gamma])&=[\overline{A}\oplus \overline{B}, \gamma (p^A_-\oplus p^B_-)\gamma^{-1}]-[\overline{A}\oplus \overline{B},  p^A_-\oplus p^B_-] \notag\\
    &=[\overline{A}\oplus \overline{B},   (\alpha^{-1}\oplus \beta^{-1})^{-1} (p^A_-\oplus p^B_-) (\alpha^{-1}\oplus \beta^{-1})]-[\overline{A}\oplus \overline{B},  p^A_-\oplus p^B_-] \notag \\
    &=\phi([A, \alpha])+\phi([B, \beta]).
\end{align}
By definition $\phi([A, 1])=0$ and for two automorphisms $\alpha, \alpha'$ of $A$
$$\phi([A, \alpha\alpha'])=\phi([A\oplus A, \alpha\alpha'\oplus 1])=\phi([A\oplus A, \alpha\oplus \alpha'])= \phi([A, \alpha])+ \phi([A,\alpha']).$$
\end{proof}

\begin{lem}
    The map $\phi$ is surjective.
\end{lem}
\begin{proof}
    Let $B$ be an object in $\mathcal{C}_i(\CA)$ and $p: B\rightarrow B$ a projection. Define $A(j_1, \dots, j_{i+1})= B(j_1, \dots, j_i)$ for every $j_{i+1}$ and $\alpha: A\rightarrow A$ by

\[ 
\begin{tikzcd}[column sep=2.5em, row sep=2.5em]
A \arrow[d, "\al="] & \cdots \arrow[dr, "p"] & B \arrow[d, "1-p"] \arrow[dr, "p"] & B \arrow[d, "1-p"] \arrow[dr, "p"] & B \arrow[d, "1-p"] \arrow[dr, "p"] & B \arrow[d, "1-p"] \arrow[dr, "p"] &\cdots   \\
A          &  \cdots                 & B  & B  & B & B  &\cdots
\end{tikzcd}
\]

We compute $\al p^A_- \al^{-1}$ and get

\[ 
\begin{tikzcd}[column sep=2.5em, row sep=2.5em]
A \arrow[d, "\al p^A_- \al^{-1}="] & \cdots  & B \arrow[d, "1"] & B \arrow[d, "1"]  & B \arrow[d, "p"]  & B \arrow[d, "0"] & B \arrow[d, "0"] &\cdots   \\
A          &  \cdots                 & B  & B  & B & B &B  &\cdots
\end{tikzcd}
\]

Therefore, $\phi([A, \al])=[B^{\oplus 5}, 1\oplus 1 \oplus p \oplus 0 \oplus 0]=[B, p].$

\end{proof}
\begin{defn}
    Let $\alpha$ be an automorphism of $A$ in $\mathcal C_{i+1}(\CA)$. We say $\alpha$ is split at $m$ if the following holds:
    $$j_{i+1}\geq m \text{ implies } \alpha(A(j_1, \dots, j_{i+1}))\subset \bigoplus_{(k_1, \dots, k_{i+1}): k_{i+1}\geq m}A(k_1, \dots, k_{i+1})$$
    and 
    $$j_{i+1}< m \text{ implies } \alpha(A(j_1, \dots, j_{i+1}))\subset \bigoplus_{(k_1, \dots, k_{i+1}): k_{i+1}< m}A(k_1, \dots, k_{i+1}).$$
\end{defn}
\begin{lem}
    A split automorphism is trivial in $K_1(\mathcal C_{i+1}(\CA))$.
\end{lem}
\begin{proof}
    Given $[A, \alpha]$ with $\alpha$ split at $m$. Decompose $A=A'\oplus A''$ where 
    $$A'(j_1, \dots, j_{i+1}):=\begin{cases}
        A(j_1, \dots, j_{i+1}), & j_{i+1} \geq m\\
        0, & j_{i+1}<m.
    \end{cases}$$
    Then $\alpha$ restricts to automorphisms $\alpha': A'\rightarrow A'$ and $\alpha'': A''\rightarrow A''$ with $[A, \al]=[A', \alpha']+[A'', \alpha'']$. 

    Define $s^lA' (j_1, \dots, j_{i+1})=A' (j_1, \dots, j_{i+1}-l)$. Conjugating $\al'$ by $s^l$ gives an automorphism on $s^lA'$ we denote by $(\al')_l.$ Apparently, $[s^{l}A', (\al')_{l}]=[s^{l+1}A', (\al')_{l+1}]$ by the obvious isomorphism. Therefore, 
    $$ [A', \al']+\left[\bigoplus_{l=1}^\infty s^l A', \bigoplus_{l=1}^\infty (\al')_l\right]=\left[\bigoplus_{l=0}^\infty s^l A', \bigoplus_{l=0}^\infty (\al')_l\right] = \left[\bigoplus_{l=1}^\infty s^l A', \bigoplus_{l=1}^\infty (\al')_l\right],$$
    which implies $[A', \alpha']=0.$ Note the infinite direct sums above are well defined because over every point $(j_1, \dots, j_{i+1})$, there are finitely many nonzero summands. A similar argument shows that $[A'', \alpha'']=0.$
\end{proof}
\begin{lem}
    Let $[A, p_1], [A, p_2]\in K_0(\operatorname{Kar}(\mathcal{C}_{i}(\CA)))$. Then $[A, p_1]=[A, p_2]$ if and only if there are objects $A'$ and $A''$ in $\mathcal{C}_{i}(\CA)$ and an automorphism $\phi$ of $A\oplus A'\oplus A''$ such that $(p_2\oplus 1\oplus 0)  \phi=\phi   (p_1\oplus 1\oplus 0)$.
\end{lem}
\begin{proof}
    The if direction is trivial. Assume $[A, p_1]=[A, p_2]$, which implies, for some object $[A', q]$, \((A \oplus A', p_1 \oplus q)\) is isomorphic to 
\((A \oplus A', p_2 \oplus q)\). But then \((A \oplus A' \oplus A', p_1 \oplus q \oplus (1 - q))\) is isomorphic to 
\((A \oplus A' \oplus A', p_2 \oplus q \oplus (1 - q))\). Conjugating \((A' \oplus A', q \oplus (1 - q))\) by 
\[
\ \begin{pmatrix}
q & 1 - q \\
1 - q & q
\end{pmatrix} 
\]
gives \((A' \oplus A', 1 \oplus 0)\) so we obtain the desired result by letting \(A'' = A'\). 

\end{proof}
\begin{lem}
    The map $\phi$ is injective.
\end{lem}
\begin{proof}

Assume \( \phi([A, \alpha]) = 0 \) for some $A\in \mathcal{C}_{i+1}(\CA)$ and automorphism $\alpha$. We have
\[
[\overline{A}, \alpha p^A_- \alpha^{-1}] - [\overline{A}, p^A_-] = 0.
\]
    Find \( A' \) and \( A'' \) in \( \mathcal{C}_i(\CA) \) such that
\begin{equation} \label{eq: LHS}
    (\overline{A} \oplus A' \oplus A'', p^A_- \oplus 1 \oplus 0)
\end{equation}

is isomorphic to
\begin{equation}\label{eq: RHS}
    (\overline{A} \oplus A' \oplus A'', \alpha p^A_- \alpha^{-1} \oplus 1 \oplus 0) = (\overline{A} \oplus A' \oplus A'', (\alpha \oplus 1 \oplus 1)(p^A_- \oplus 1 \oplus 0)(\alpha \oplus 1 \oplus 1)^{-1}).
\end{equation}

Define $ (C, \gamma)= (A \oplus B, \alpha \oplus 1) $ where
\[
B(j_1, \ldots, j_{i+1}) =
\begin{cases}
A'(j_1, \ldots, j_i ) & j_{i+1} = -1 \\
A''(j_1, \ldots, j_i ) & j_{i+1} = 0\\
0 & \text{otherwise}.
\end{cases}
\]

Then expression~\eqref{eq: LHS} is equal to $[\overline{C}, p_-^C]$ and expression~\eqref{eq: RHS} is equal to $[\overline{C}, \gamma p_-^C\gamma^{-1}]$. Thus, there is an isomorphism \( \beta : \overline{C} \to \overline{C} \) satisfying
$$
\beta \gamma p_-^C\gamma^{-1} = p^C_- \beta.
$$
Since $\overline{C}$ is obtained from summing terms in a band of $C$, one could lift $\beta$ to an automorphism of this band which agrees with $\beta$ when descending
to $\overline{C}.$ Extending this automorphism from the band to all $C$ by the identity map outside the band, we get an automorphism $\tilde \beta$ of $C$ that 
\[
\tilde\beta \gamma p^C_- = p^C_- \tilde\beta \gamma.
\]
This means that \( \tilde\beta \gamma \) is split at \( 0 \), so \( [C, \tilde\beta \gamma] = 0 \). However \( \tilde\beta \) is the identity outside a finite band, so \( \tilde\beta \) is split. Hence \( [C, \tilde\beta] = 0 \) and \( [A, \al]=[C, \gamma] = 0 \). 
\end{proof}
Conflict of interest statement: NO; Data availability statement: NO.
\printbibliography
\end{document}